\RequirePackage[hyphens]{url}
\documentclass[submission,copyright,creativecommons]{eptcs}
\providecommand{\event}{MFPS 2021} 

\let\doappendix\undefined

\usepackage{graphicx}
\usepackage{enumerate}
\usepackage{amsmath}
\usepackage{amsthm}
\usepackage{hyperref}
\usepackage{csquotes}
\usepackage{amssymb}
\usepackage{bbm}
\usepackage[utf8]{inputenc}
\usepackage{tikz}
\usepackage{tikz-cd}
\usepackage{stmaryrd}
\usepackage{cite}
\usetikzlibrary{automata,positioning,arrows}
\usepackage{mathrsfs}
\usepackage{amssymb}
\usepackage[capitalise]{cleveref}
\usepackage{subcaption}
\usepackage{graphics}
\usepackage{adjustbox}
\usepackage{hyperref}
\usepackage{cleveref}
\usepackage[textsize=tiny]{todonotes}
\usepackage[english]{babel}
\usepackage{breakurl}             
\usepackage{underscore}           

\newcommand{\Pow}{\mathcal{P}}
\newcommand{\lang}{\mathcal{L}}

\newcommand{\Set}{\textnormal{Set}}
\newcommand{\C}{\mathscr{C}}

\usetikzlibrary{shapes.geometric}

\newtheorem{theorem}{Theorem}
\newtheorem{proposition}[theorem]{Proposition}
\newtheorem{lemma}[theorem]{Lemma}
\newtheorem{corollary}[theorem]{Corollary}
\theoremstyle{definition}
\newtheorem{definition}[theorem]{Definition}

\newtheorem{example}[theorem]{Example}
\theoremstyle{remark}

\crefname{prop}{Proposition}{Propositions}
\crefname{thm}{Theorem}{Theorems}
\crefname{def}{Definition}{Definitions}
\crefname{lemma}{Lemma}{Lemmas}
\crefname{example}{Example}{Examples}
\crefname{corollary}{Corollary}{Corollaries}

\newcommand{\deftag}[1]{{\hskip-\labelsep\ \upshape\bfseries (#1)}\hskip\labelsep}
\newcommand{\tagcite}[1]{{\hskip-\labelsep\ \upshape\bfseries \cite{#1}}\hskip\labelsep}

\makeatletter
\newenvironment{oneshot}[2]
    {
        \bgroup
        \addtocounter{theorem}{-1}
        \expandafter\def\csname thetheorem\endcsname{\ref{#2}}
        \def\foo{\end{#1}}
        \begin{#1}
    }{
        \foo
        \egroup
    }
\makeatother

\title{Canonical Automata via Distributive Law Homomorphisms}
            \author{Stefan Zetzsche\thanks{The author has been supported by GCHQ via the VeTSS grant “Automated black-box verification of networking systems” (4207703/RFA 15845).} \qquad Gerco van Heerdt\thanks{The author has been supported by the EPSRC via the Standard Grant CLeVer (EP/S028641/1).}
\institute{University College London}
\email{s.zetzsche@cs.ucl.ac.uk}
\and
            Matteo Sammartino\thanks{The author has been supported by the EPSRC via the Standard Grant CLeVer (EP/S028641/1).}
\institute{Royal Holloway, University of London}
\institute{University College London}
\and
            Alexandra Silva\thanks{The author has been supported by the ERC via the Consolidator Grant AutoProbe 101002697 and by a Royal Society Wolfson Fellowship.}
\institute{Cornell University}
\institute{University College London}
}
\def\titlerunning{Canonical Automata via Distributive Law Homomorphisms}
\def\authorrunning{S. Zetzsche, G. van Heerdt, M. Sammartino, A. Silva}
\begin{document}
\maketitle

\begin{abstract}
The classical powerset construction is a standard method converting a non-deterministic automaton into a deterministic one recognising the same language. Recently, the powerset construction has been lifted to a more general framework that converts an automaton with side-effects, given by a monad, into a deterministic automaton accepting the same language. The resulting automaton has additional algebraic properties, both in the state space and transition structure, inherited from the monad.  In this paper, we study the reverse construction and present a framework in which a deterministic automaton with additional algebraic structure over a given monad can be converted into an equivalent succinct automaton with side-effects. Apart from recovering examples from the literature, such as the canonical residual finite-state automaton and the \'atomaton, we discover a new canonical automaton for a regular language by relating the free vector space monad over the two element field to the neighbourhood monad. Finally, we show that every regular language satisfying a suitable property parametric in two monads admits a size-minimal succinct acceptor.  
\end{abstract}

\section{Introduction}

The existence of a unique minimal \emph{deterministic} acceptor is an important property of regular languages. Establishing a similar result for \emph{non-deterministic} acceptors is  significantly more difficult, but nonetheless of great practical importance, as non-deterministic automata (NFA) can be exponentially more succinct than deterministic ones (DFA). The main issue is that a regular language can be accepted by several size-minimal NFAs that are not isomorphic.
A number of sub-classes of non-deterministic automata have been identified in the literature to tackle this issue, which all admit canonical representatives: the \emph{\'atomaton}~\cite{BrzozowskiT14}, the \emph{canonical residual finite-state automaton} (short \emph{canonical RFSA} and also known as \emph{jiromaton})~\cite{DenisLT02}, the \emph{minimal xor automaton}~\cite{VuilleminG210}, and the \emph{distromaton}~\cite{MyersAMU15}. 

In this paper we provide a general categorical framework that unifies constructions of canonical non-deterministic automata and unveils new ones. Our framework adopts the well-known representation of side-effects via \emph{monads}~\cite{moggi1991notions} to generalise non-determinism in automata. For instance, an NFA (without initial states) can be represented as a pair $\langle X,k \rangle $, where $X$ is the set of states and 
$
	k \colon X \to 2 \times \Pow(X)^A
$
combines the function classifying each state as accepting or rejecting with the function giving the set of next states for each input. The powerset forms a monad $\langle \Pow, \{-\}, \mu \rangle$, where $\{-\}$ creates singleton sets and $\mu$ takes the union of a set of sets. This allows describing the classical powerset construction, converting an NFA into a DFA, in categorical terms \cite{silva2010generalizing} as depicted on the left of \Cref{gen-det-diagrams},
\begin{figure*}[t]
\centering
	\begin{tikzcd}[ ampersand replacement=\&]
		X \ar{d}[swap]{k} \ar{r}{\{-\}} \&
			\mathcal{P}(X) \ar{dl}{k^\sharp} \ar[dashed]{r}{\textnormal{obs}} \&
			2^{A^*} \ar{d}{\langle \varepsilon, \delta \rangle} \\
		2 \times \mathcal{P}(X)^A \ar[dashed]{rr}[below]{2 \times \textnormal{obs}^A} \&
			\&
			2 \times (2^{A^*})^A
	\end{tikzcd}
	\qquad
			\begin{tikzcd}[ampersand replacement=\&]
			X \ar{d}[swap]{k} \ar{r}{\eta} \&
			TX \ar{dl}{k^\sharp} \ar[dashed]{r}{\textnormal{obs}} \&
				\Omega \ar{d}{\omega} \\
			FTX \ar[dashed]{rr}[below]{F\textnormal{obs}} \&
				\&
				F\Omega
		\end{tikzcd}.
\caption{Generalised determinisation of automata with side-effects in a monad.}
\label{gen-det-diagrams}
\end{figure*}
where $k^\sharp \colon \Pow(X) \to 2 \times \Pow(X)^A$ represents an equivalent DFA, obtained by taking the subsets of $X$ as states, and $\langle \varepsilon, \delta \rangle : 2^{A^*} \rightarrow 2 \times (2^{A^*})^A$ is the automaton of languages. There then exists a unique automaton homomorphism $\textnormal{obs}$, assigning a language semantics to each set of states.

As seen on the right of \Cref{gen-det-diagrams} this perspective further enables a \emph{generalised determinisation} construction \cite{silva2010generalizing}, where $2 \times (-)^A$ is replaced by any (suitable) functor $F$ describing the automaton structure, and $\mathcal{P}$ by a monad $T$ describing the automaton side-effects.
$\Omega \xrightarrow{\omega} F \Omega$ is the so-called \emph{final coalgebra}, providing a semantic universe that generalises the automaton of languages.

Our work starts from the observation that the deterministic automata resulting from this generalised determinisation constructions have \emph{additional algebraic structure}: the state space $\Pow(X)$ of the determinised automaton defines a free complete join-semilattice (CSL) over $X$, and $k^\sharp$ and $\textnormal{obs}$ are CSL homomorphisms. More generally, $TX$ defines a (free) algebra for the monad $T$, and $k^\sharp$ and $\textnormal{obs}$ are $T$-algebra homomorphisms.

With this observation in mind, our question is: can we exploit the additional algebraic structure to ``reverse'' these constructions? In other words, can we convert a deterministic automaton with additional algebraic structure over a given monad to an equivalent succinct automaton with side-effects, possibly over another monad? To answer this question, the paper makes the following contributions:
\begin{itemize}
	\item We present a general categorical framework based on bialgebras and distributive law homomorphisms that allows deriving canonical representatives for a wide class of succinct automata with side-effects in a monad. 		
	\item We strictly improve the expressivity of previous work \cite{HeerdtMSS19, arbib1975fuzzy}: our framework instantiates not only to well-known examples such as the canonical RFSA (\Cref{canonicalrfsaexample}) and the minimal xor automaton (\Cref{minimalxorexmaple}), but also includes the \'atomaton (\Cref{atomatonexample}) and the distromaton (\Cref{distromatonexample}), which were not covered in \cite{HeerdtMSS19, arbib1975fuzzy}.
 While other frameworks restrict themselves to the category of sets \cite{HeerdtMSS19}, we are able to include canonical acceptors in other categories, such as the \textit{canonical nominal RFSA} (\Cref{nominalexample}). 
	\item We relate vector spaces over the unique two element field with complete atomic Boolean algebras and consequently discover a previously unknown canonical mod-2 weighted acceptor for regular languages---the \emph{minimal xor-CABA automaton} (\Cref{minimalxorcabaexample})---that in some sense is to the minimal xor automaton what the \'atomaton is to the canonical RFSA  (\Cref{minimalxorcabadiagram}).
	\item We introduce an abstract notion of \emph{closedness} for succinct automata that is parametric in two monads (\Cref{closedsuccinctdef}), and 
	prove that every regular language satisfying a suitable property admits a canonical size-minimal representative among closed acceptors (\Cref{minimalitytheorem}). By instantiating the latter we subsume known minimality results for canonical automata, prove the xor-CABA automaton minimal, and establish a size comparison between different acceptors (\Cref{minimalityimplications}). 
\end{itemize}

\ifdefined\doappendix\else An extended version of this paper is available at \cite{arxiv}.\fi

\section{Overview of the approach}
\label{overview}
	
In this section, we give an overview of the ideas of the paper through an example. We show how our methodology allows recovering the construction of the 
\'atomaton for 
the regular language
$
	\mathcal{L} = (a+b)^*a
$, which consists of all words over $A = \lbrace a, b \rbrace$ that end in $a$. For each step, we hint at how it is generalised in our framework.

The classical construction of the \'atomaton for $\lang$ consists in closing the \emph{residuals}\footnote{A language is a \textit{residual} or \textit{left quotient} of $\mathcal{L} \subseteq A^*$, if it is of the form $v^{-1}\mathcal{L} = \lbrace u \in A^* \mid vu \in \mathcal{L} \rbrace$ for some $v \in A^*$. } of $\lang$ under all Boolean operations, and then forming a non-deterministic automaton whose states are the \emph{atoms}\footnote{A non-zero element $a \in B$ is called \emph{atom}, if for all $x \in B$ such that $x \leq a$ one finds $x = 0$ or $x = a$.}  of the ensuing complete atomic Boolean algebra (CABA)---that is, non-empty intersections of complemented or uncomplemented residuals.
In our categorical setting, this construction is obtained in several steps, which we now describe.

\subsection{Computing residuals}
We first construct the minimal DFA accepting $\lang$ as a coalgebra of type 
	$
		 M_\lang \to 2 \times (M_\lang)^{A} \enspace
	$.
	By the well-known Myhill-Nerode theorem~\cite{nerode1958linear}, $M_\lang$ is the set of residuals for $\lang$. The automaton is depicted in \Cref{m(l)}.

	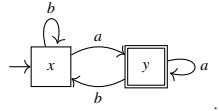
\begin{figure*}
		\tiny
		\center
	\begin{tikzpicture}[node distance=6em]
	\node[state, shape=rectangle, initial, initial text=] (x) {$x$};	
		\node[state, shape=rectangle, right of=x, accepting] (y) {$y$};
	    \path[->]
	(x) edge[loop above] node{$b$} (x)
	(y) edge[loop right] node{$a$} (y)
	(x) edge[above, bend left] node{$a$} (y)
	(y) edge[below, bend left] node{$b$} (x)
	;
	\end{tikzpicture}.
		\caption{The minimal DFA for $\mathcal{L} = (a+b)^*a$.}
	\label{m(l)}
	\end{figure*}
			In our framework, we consider coalgebras over an arbitrary endofunctor $F \colon \C \to \C$ ($F = 2 \times (-)^{A}$ and $\C = \Set$ in this case).
Minimal realisations, generalising minimal DFAs, exist for a wide class of functors $F$ and categories $\C$, including all the examples in this paper.

\subsection{Taking the Boolean closure}
We close the minimal DFA under all Boolean operations, generating an equivalent deterministic automaton that has additional algebraic structure: its state space is a CABA.
This is achieved via a double powerset construction---where sets of sets are interpreted as full disjunctive normal form---and the resulting coalgebra is of type
	$
		\Pow^2(M_\lang) \to 2 \times (\Pow^2(M_\lang))^{A}	$.
	Our construction relies on the so-called \emph{neighbourhood monad} $\mathcal{H}$, whose algebras are precisely CABAs, and yields a (free) \emph{bialgebra} capturing both the coalgebraic and the algebraic structure; the interplay of these two structures is captured via a \emph{distributive law}. 
 We then minimise this DFA to identify Boolean expressions evaluating to the same language. As desired, the resulting state space is precisely the Boolean closure of the residuals of $\lang$. Formally, we obtain the minimal bialgebra for $\lang$ depicted in \Cref{overlinem(l)atom}.

This step in our framework is generalised as closure of an $F$-coalgebra w.r.t\ (the algebraic structured induced by) any monad $S$ for which a suitable distributive law $\lambda$ with the coalgebra endofunctor $F$ exists. The first step of the closure yields a free $\lambda$-bialgebra, comprised of both an $F$-coalgebra and an $S$-algebra over the same state space. In a second step, minimisation is carried out in the category of $\lambda$-bialgebras, which guarantees simultaneous preservation of the algebraic structure and of the language semantics.

\subsection{Constructing the \'atomaton}
This step is the key technical result of our paper. Atoms have the property that their Boolean closure generates the entire CABA. In our framework, this property is generalised via the notion of \emph{generators} for algebras over a monad, which allows one to represent a bialgebra as an equivalent \emph{free} bialgebra over its generators, and hence to obtain succinct canonical representations (\Cref{forgenerator-isharp-is-bialgebra-hom}). In \Cref{succinctbialgebra1} we apply this result to obtain the canonical RFSA, the canonical nominal RFSA, and the minimal xor automaton for a given regular language.

However, to recover the \'atomaton from the minimal CABA-structured DFA of the previous step, in addition a subtle change of perspective is required. In fact, we are still working with the ``wrong'' side-effect: the non-determinism of bialgebras so far is determined by $\mathcal{H}$, whereas we are interested in an NFA, whose non-determinism is captured by the powerset monad $\mathcal{P}$. 
As is well-known, every element of a CABA can be obtained as the join of the atoms below it.
In other words, those atoms are also generators of the underlying CSL, which is an algebra for $\mathcal{P}$. We formally capture this idea as a map between monads $\mathcal{H} \to \mathcal{P}$. Crucially, we show that this map lifts to a \emph{distributive law homomorphism} and allows translating a bialgebra over $\mathcal{H}$ to a bialgebra over $\mathcal{P}$, which can be represented as a free bialgebra over atoms---the \'atomaton for $\lang$, which is shown in \Cref{atomaton}.

In \Cref{succinctbialgebra2} we generalise this idea to the situation of two monads $S$ and $T$ involved in distributive laws with the coalgebra endofunctor $F$. In particular, \Cref{generatorbialgebrahom} is our free representation result, spelling out a condition under which a bialgebra over $S$ can be represented as a free bialgebra over $T$, and hence admits an equivalent succinct representation as an automaton with side-effects in $T$. Besides the \'atomaton and the examples in \Cref{succinctbialgebra1}, this construction allows us to capture the distromaton and a newly discovered canonical acceptor that relates CABAs with vector spaces over the two element field.

\section{Preliminaries}

\begin{figure*}
\center	
	 \tiny
         \adjustbox{valign=m}{\begin{tikzpicture}[node distance= 8.5em]
	\node[state, initial, shape=rectangle, initial text=] (1) {$\lbrack \lbrace \lbrace x \rbrace, \lbrace x, y \rbrace \rbrace \rbrack $};

	\node[state, shape=rectangle, right of = 1] (3) {$\lbrack \emptyset \rbrack$};

	\node[state, shape=rectangle,below of = 1, accepting] (5) {$\lbrack\lbrace \lbrace x \rbrace,  \lbrace y \rbrace, \lbrace x, y \rbrace \rbrace \rbrack$};

	\node[state, shape=rectangle,right of = 5, accepting] (7) {$\lbrack\lbrace \lbrace y \rbrace \rbrace \rbrack$};

	\node[state, shape=rectangle,right of = 3] (9) {$\lbrack\lbrace \emptyset \rbrace \rbrack$};

	\node[state, shape=rectangle,right of = 9] (11) {$\lbrack \lbrace \lbrace x, y \rbrace, \emptyset \rbrace \rbrack$};

	\node[state, shape=rectangle,right of = 7, accepting] (13) {$\lbrack \lbrace \lbrace y \rbrace, \emptyset \rbrace \rbrack$};

	\node[state, shape=rectangle, right of = 13, accepting] (15) {$\lbrack \lbrace \lbrace x \rbrace, \lbrace y \rbrace, \lbrace x, y \rbrace, \emptyset \rbrace \rbrack$};
			
		    \path[->]
	(3) edge[loop above] node{$a,b$} (3)
	(1) edge[loop above] node{$b$} (1)
	(1) edge[right, bend left] node{$a$} (5)
	(7) edge[left] node{$a,b$} (3)
	(5) edge[loop below] node{$a$} (5)
	(5) edge[left, bend left] node{$b$} (1)
	(9) edge[bend left, right] node{$b$} (13)
	(9) edge[loop above] node{$a$} (9)
	(11) edge[left] node{$a,b$} (15)
	(13) edge[left, bend left] node{$a$} (9)
	(13) edge[loop below] node{$b$} (13)
	(15) edge[loop below] node{$a,b$} (15)
;
         \end{tikzpicture}}
	\qquad
                \adjustbox{valign=m}{\resizebox{0.35 \columnwidth}{!}{%
	\begin{tabular}[]{ c|c|c|c|c|c|c|c|c } 
		$\wedge$ & $1$ & $2$ & $3$ & $4$ & $5$ & $6$ & $7$ & $8$ \\
		\hline
		$1$ & $1$ & $2$ & $2$ & $1$ & $1$ & $2$ & $2$ & $1$ \\
		\hline
		$2$ & $2$ & $2$ & $2$ & $2$  &  $2$ & $2$  & $2$  & $2$  \\
		\hline
		$3$ & $2$ & $2$ & $3$ & $3$ & $2$ & $2$ & $3$ & $3$ \\
		\hline
		$4$ & $1$ & $2$ & $3$ & $4$ & $1$ & $2$ & $3$ & $4$ \\
		\hline
		$5$ & $1$ & $2$ & $2$ & $1$ & $5$ & $6$ & $6$ & $5$ \\
		\hline
		$6$ & $2$ & $2$ & $2$ & $2$ & $6$ & $6$ & $6$ & $6$ \\
		\hline
		$7$ & $2$ & $2$ & $3$ & $3$ & $6$ & $6$ & $7$ & $7$ \\
		\hline
		$8$ & $1$ & $2$ & $3$ & $4$ & $5$ & $6$ & $7$ & $8$
	\end{tabular}
	\qquad
		\begin{tabular}[]{ c|c} 
		  & $\neg$\\
		 \hline
		 $1$ & $7$ \\
		 \hline
		 $2$ & $8$ \\
		 \hline
		 $3$ & $5$ \\
		 \hline
		 $4$ & $6$ \\
		 \hline
		 $5$ & $3$ \\
		 \hline
		 $6$ & $4$ \\
		 \hline
		 $7$ & $1$ \\
		 \hline
		 $8$ & $2$ 
		\end{tabular}
        }}
		\caption{The minimal CABA-structured DFA for $\mathcal{L} = (a+b)^*a$, where
$1 \equiv \lbrack \lbrace \lbrace x \rbrace, \lbrace x, y \rbrace \rbrace \rbrack$, $2 \equiv \lbrack \emptyset \rbrack$, $3 \equiv \lbrack \lbrace \emptyset \rbrace \rbrack$, $4 \equiv \lbrack \lbrace \lbrace x, y \rbrace, \emptyset \rbrace \rbrack$, $5 \equiv \lbrack \lbrace \lbrace x \rbrace, \lbrace y \rbrace, \lbrace x, y \rbrace \rbrace \rbrack$, $6 \equiv \lbrack \lbrace \lbrace y \rbrace \rbrace \rbrack$, $7 \equiv \lbrack \lbrace \lbrace y \rbrace, \emptyset \rbrace \rbrack$, $8 \equiv \lbrack \lbrace \lbrace x \rbrace, \lbrace y \rbrace, \lbrace x, y \rbrace, \emptyset \rbrace \rbrack$.}
	\label{overlinem(l)atom}
	\end{figure*}
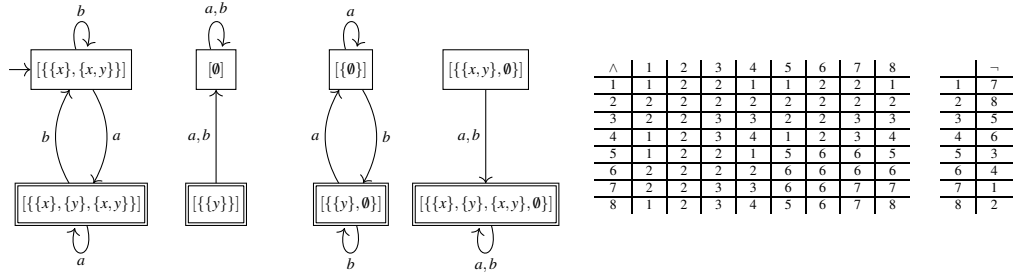 

\label{preliminaries}

We assume basic knowledge of category theory (including functors, natural transformations, and adjunctions)~\cite{awodey2010category}. In this section we recall the relevant notions for our technical development: coalgebras, monads, algebras over a monad, distributive laws, and bialgebras.

Unpointed deterministic automata are basic examples of coalgebras in the category of sets and functions: they are of the type $k \colon X \to FX$, where $FX = 2 \times X^A$ and $k$ pairs the final state function and the transition function assigning a next state to each letter $a\in A$.
Coalgebra has emerged as a unifying framework to study infinite data types and state-based systems \cite{rutten2000universal}.
\begin{definition}\deftag{Coalgebra}
	A \emph{coalgebra} for an endofunctor $F$ in a category $\C$ is a tuple $\langle X, k \rangle$ consisting of an object $X$ in $\C$ and a morphism $k\colon X \rightarrow FX$.
\end{definition}
Crucial in the theory of coalgebras is the notion of homomorphism, which allows to relate states of coalgebras of the same type. A homomorphism $f\colon \langle X, k_X \rangle \rightarrow \langle Y,k_Y \rangle$ between $F$-coalgebras is a morphism $f\colon X \rightarrow Y$ satisfying  $k_Y \circ f = Ff \circ k_X$. The category of $F$-coalgebras and homomorphisms is denoted by $\textnormal{Coalg}(F)$. If it exists, the final object of this category is of particular importance.
\begin{definition}\deftag{Final coalgebra}
	An $F$-coalgebra $\langle \Omega, k_{\Omega} \rangle$ is \emph{final} if every $F$-coalgebra $\langle X, k \rangle$ admits a unique homomorphism $\textnormal{obs}_{\langle X, k \rangle}: \langle X, k \rangle \rightarrow \langle \Omega, k_{\Omega} \rangle$.
\end{definition}
 The unique final coalgebra homomorphism can be understood as the observable behaviour of a system. For example, for the functor $FX = 2 \times X^A$, the final $F$-coalgebra is the set of all languages $\Pow(A^\star)$ and the final coalgebra homomorphism assigns to a state $x$ of an unpointed deterministic automaton the language in $\Pow(A^*)$ it accepts\footnote{For a deterministic automaton given by $\varepsilon: X \rightarrow 2$ and $\delta: X \rightarrow X^A$, acceptance is coinductively defined as a function $\textnormal{obs}: X \rightarrow 2^{A^*}$ by $\textnormal{obs}(x)(\varepsilon) = \varepsilon(x)$ and $\textnormal{obs}(x)(av) = \textnormal{obs}(\delta(x)(a))(v)$.} when given the initial state $x$.

In the context of computer science, monads have been introduced by Moggi as a general perspective on exceptions, side-effects, non-determinism, and continuations \cite{moggi1988computational, moggi1990abstract,		moggi1991notions}. 
\begin{definition}\deftag{Monad}
	A \emph{monad} on a category $\mathscr{C}$ is a tuple $\langle T, \eta, \mu \rangle$ consisting of an endofunctor $T: \mathscr{C} \rightarrow \mathscr{C}$ and natural transformations
	$
	\eta: \textnormal{id}_{\mathscr{C}} \Rightarrow T$ and $\mu: T^2 \Rightarrow T$
	satisfying $\mu \circ T\mu = \mu \circ \mu_T$ and $\mu \circ \eta_T = \textnormal{id}_T = \mu \circ T\eta$.
\end{definition}
By a slight abuse of notation we will refer to a monad simply by its underlying endofunctor. 

	Non-determinism is typically modelled by the \emph{powerset monad} $\mathcal{P} $, whose underlying endofunctor $\mathcal{P}$ assigns to a set $X$ the set of subsets $\mathcal{P}X$; whose unit maps an element $x$ to the singleton $\eta_X(x) = \lbrace x \rbrace$; and whose multiplication flattens subsets by taking their union $\mu_X(\Phi) = \bigcup_{U \in \Phi} U$. Other monads that play a role for us are the \emph{nominal powerset monad} $ \mathcal{P}_{\textnormal{n}}$ \cite{moerman2019residual}, the \emph{neighbourhood monad}
 $\mathcal{H}$ \cite{jacobs2015recipe}, the \emph{monotone neighbourhood monad}
 $\mathcal{A}$ \cite{jacobs2015recipe}, and the \emph{free vector space monad}
 $\mathcal{R}$ over the unique two element field \cite{jacobs2011bases}. The formal definitions are given in \ifdefined\doappendix\Cref{monaddefs}\else \cite[Definition~40]{arxiv}\fi.

The concept of a monad can also be seen as an alternative to Lawvere theory as a category theoretic formulation of universal algebra \cite{eilenberg1965adjoint, linton1966some}.  
\begin{definition}\deftag{Algebra over a monad}
	An \emph{algebra} over a monad $T$ on $\C$ is a tuple $\langle X, h \rangle$ consisting of an object $X$ in $\C$ and a morphism $h: TX \rightarrow X$ satisfying $h \circ \mu_X = h \circ Th$ and $h \circ \eta_X = \textnormal{id}_X$.
\end{definition}  Every object admits a \emph{free} algebra $\langle TX, \mu_X \rangle$. A homomorphism $f: \langle X, h_X \rangle \rightarrow \langle  Y, h_Y \rangle$ between $T$-algebras is a morphism $f: X \rightarrow Y$ satisfying $h_Y \circ Tf = f \circ h_X$. The category of $T$-algebras and homomorphisms is denoted by $\textnormal{Alg}(T)$.

\begin{example}
\begin{itemize}
	\item The category $\textnormal{Alg}(\mathcal{P})$ is isomorphic to the category of complete join-semi lattices (CSL) and functions that preserve all joins \cite{jacobs2011bases}.
	\item The category $\textnormal{Alg}(\mathcal{H})$ is isomorphic to the category of complete atomic Boolean algebras (CABA) and Boolean algebra homomorphisms that preserve all meets and all joins \cite{jacobs2015recipe}.
	\item The category $\textnormal{Alg}(\mathcal{A})$ is isomorphic to the category of completely distributive lattices (CDL) and functions that preserve all meets and all joins \cite{jacobs2015recipe}.
	\item The category $\textnormal{Alg}(\mathcal{R})$ is isomorphic to the category of vector spaces over the unique two element field ($\mathbb{Z}_2$-Vect) and linear maps \cite{jacobs2011bases}.
\end{itemize}
	
\end{example}

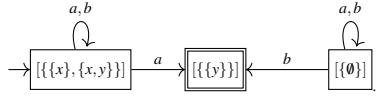
\begin{figure*}
	\center
	\tiny
\begin{tikzpicture}[node distance= 8.5em]
	\node[state, initial, shape=rectangle, initial text=] (1) {$\lbrack \lbrace \lbrace x \rbrace, \lbrace x, y \rbrace \rbrace \rbrack $};

	\node[state, shape=rectangle,right of = 1, accepting] (2) {$\lbrack\lbrace \lbrace y \rbrace \rbrace \rbrack$};

	\node[state, shape=rectangle,right of = 2] (3) {$\lbrack\lbrace \emptyset \rbrace \rbrack$};

		    \path[->]
	(1) edge[loop above] node{$a,b$} (1)
	(1) edge[above] node{$a$} (2)
	(3) edge[loop above] node{$a,b$} (3)
	(3) edge[above] node{$b$} (2)	
;
	\end{tikzpicture}.
	\caption{The \'atomaton for $\mathcal{L} = (a+b)^*a$.}
\label{atomaton}
\end{figure*}

Distributive laws have originally occurred as a way to compose monads \cite{beck1969distributive}, but now also exist in a wide range of other forms \cite{Street2009}. For our particular case it is sufficient to consider distributive laws between a monad and an endofunctor, sometimes referred to as \emph{Eilenberg-Moore laws}\cite{jacobs2015trace}.
\begin{definition}\deftag{Distributive law}
	A \emph{distributive law} between a monad $T$  and an endofunctor $F$ on $\C$ is a natural transformation $\lambda: TF \Rightarrow FT$
satisfying $F \eta_X = \lambda_X \circ \eta_{FX}$ and $\lambda_X \circ \mu_{FX} = F\mu_X \circ \lambda_{TX} \circ T{\lambda_X}$.
\end{definition}
 For example, every algebra $h: TB \rightarrow B$ for a set monad $T$ induces a  distributive law $\lambda^h$ between $T$ and $F$ with $FX = B \times X^A$ defined by 
 \begin{equation}
 \label{induceddistrlaweq}
 	\lambda^{h}_X := (h \times \textnormal{st}) \circ \langle T\pi_1, T\pi_2 \rangle,
 \end{equation}
 where $\textnormal{st}$ denotes the usual strength function\footnote{For any two sets $X,A$ the strength function $\textnormal{st}: T(X^A) \rightarrow (TX)^A$ is defined by $ \textnormal{st}(U)(a) = T(\textnormal{ev}_a)(U)$,
	where $\textnormal{ev}_{a}(f) = f(a)$.} \cite{jacobs2005bialgebraic}. We are particularly interested in canonical algebra structures for the output set $B = 2$. For instance, the algebra structures defined by $h^{\mathcal{P}}(\varphi) = h^{\mathcal{R}}(\varphi) =  \varphi(1)$ and $h^{\mathcal{H}}(\Phi) = h^{\mathcal{A}}(\Phi) = \Phi(\textnormal{id}_2)$, where we identify subsets with their characteristic functions. In these cases we will abuse notation and write $\lambda^T$ instead of $\lambda^{h^T}$. 

\begin{example}\deftag{Generalized determinisation \cite{rutten2013generalizing}}
\label[example]{determinisationexample}
Given a distributive law, one can model the determinisation of a system with dynamics in $F$ and side-effects in $T$ (sometimes referred to as \emph{succinct} automaton) by lifting a $FT$-coalgebra $\langle X, k \rangle$ to the $F$-coalgebra 
	$\langle TX, k^{\sharp} \rangle$, where $k^{\sharp} := (F \mu_X \circ \lambda_{TX}) \circ Tk$. As one verifies, the latter is in fact a $T$-algebra homomorphism of type $k^{\sharp}: \langle TX, \mu_X \rangle \rightarrow \langle FTX, F\mu_X \circ \lambda_{TX} \rangle$. For instance, if the distributive law $\lambda$ is induced by the  disjunctive $\mathcal{P}$-algebra $h^{\mathcal{P}}: \mathcal{P}2 \rightarrow 2$ with $h^{\mathcal{P}}(\varphi) = \bigvee_{u \in \varphi} u = \varphi(1)$, the lifting $k^{\sharp}$ is the DFA in CSL obtained from an NFA $k$ via the classical powerset construction.   
\end{example}	
	
	The example above illustrates the concept of a bialgebra: the algebraic part $(TX, \mu_X)$ and the coalgebraic part $(TX, k^{\sharp})$ of a lifted automaton are compatible along the distributive law $\lambda$.

    \begin{definition}\deftag{Bialgebra}
		A $\lambda$\emph{-bialgebra} is a tuple $\langle X, h, k \rangle$ consisting of a $T$-algebra $\langle X,h \rangle$ and an $F$-coalgebra $\langle X, k \rangle$ 
satisfying $Fh \circ \lambda_X \circ Tk = k \circ h$.
	\end{definition}
	 A homomorphism between $\lambda$-bialgebras is a morphism between the underlying objects that is simultaneously a $T$-algebra homomorphism and an $F$-coalgebra homomorphism.
	The category of $\lambda$-bialgebras and homomorphisms is denoted by $\textnormal{Bialg}(\lambda)$. The existence of a final $F$-coalgebra is equivalent to the existence of a final $\lambda$-bialgebra, as the next result shows.
	
	\begin{lemma}\tagcite{jacobs2012trace}
		Let $\langle \Omega, k_{\Omega} \rangle$ be the final $F$-coalgebra, then $\langle \Omega, h_{\Omega}, k_{\Omega} \rangle$ with $h_{\Omega}:= \textnormal{obs}_{\langle T\Omega, \lambda_{\Omega} \circ T k_{\Theta} \rangle}$ is the final $\lambda$-bialgebra satisfying $\textnormal{obs}_{\langle X, h, k \rangle} = \textnormal{obs}_{\langle X, k \rangle}$. Conversely, if $\langle \Omega, h_{\Omega}, k_{\Omega} \rangle$ is the final $\lambda$-bialgebra, then $\langle \Omega, k_{\Omega} \rangle$ is the final $F$-coalgebra.
	\end{lemma}
	For instance, for the distributive law in \Cref{determinisationexample},
	the final bialgebra is carried by the final coalgebra $\Pow(A^*)$ and also has a free $\mathcal{P}$-algebra structure that takes the union of languages.
	
	The generalized determinisation procedure in \Cref{determinisationexample} can now be rephrased in terms of a functor between the category of coalgebras with dynamics in $F$ and side-effects in $T$ on the one side, and the category of bialgebras on the other side.
	
	\begin{lemma}\tagcite{jacobs2012trace}
	\label[lemma]{expfunctor}
		Defining $\textnormal{exp}_T(\langle X, k \rangle) := \langle TX, \mu_X, (F \mu_X \circ \lambda_{TX}) \circ Tk \rangle$ and $\textnormal{exp}_T(f) := Tf$ yields a functor $\textnormal{exp}_T: \textnormal{Coalg}(FT) \rightarrow \textnormal{Bialg}(\lambda)$. 
	\end{lemma}
	
	We will sometimes refer to the functor which arises from the one above by precomposition with the canonical embedding of $F$-coalgebras into $FT$-coalgebras. 

	\begin{corollary}
		Defining $\textnormal{free}_T(\langle X, k \rangle) := \langle TX, \mu_X, \lambda_X \circ Tk \rangle$ and $\textnormal{free}_T(f) := Tf$ yields a functor $\textnormal{free}_T: \textnormal{Coalg}(F) \rightarrow \textnormal{Bialg}(\lambda)$ satisfying $\textnormal{free}_T(\langle X, k \rangle) = \textnormal{exp}_T(\langle X, F\eta_X \circ k\rangle)$.
	\end{corollary}

\section{Succinct automata from bialgebras}
\label{succinctbialgebra1}

In this section we introduce the foundations of our theoretical contributions.
We begin with the notion of a \textit{generator} \cite{arbib1975fuzzy} for an algebra over a monad and demonstrate how it can be used to translate a bialgebra into an equivalent free bialgebra.
While the treatment is very general, we are particularly interested in the case in which the bialgebra is given by a deterministic automaton that has additional algebraic structure over a given monad, and the translation results in an automaton with side-effects in that monad.
We will demonstrate that the theory in this section instantiates to the canonical RFSA \cite{DenisLT02}, the canonical \emph{nominal} RFSA \cite{moerman2019residual}, and the minimal xor automaton \cite{VuilleminG210}.

\begin{definition}\deftag{Generator and basis}
	A \emph{generator} for a $T$-algebra $\langle X, h \rangle$ is a tuple $\langle Y, i, d \rangle$  consisting of an object $Y$, a morphism $i \colon Y \rightarrow X$, and a morphism $d \colon X \rightarrow TY$ such that $(h \circ Ti) \circ d = \textnormal{id}_X$.
	A generator is called a \emph{basis} if it additionally satisfies $d \circ (h \circ Ti) = \textnormal{id}_{TY}$.
\end{definition}

A generator for an algebra is called a \textit{scoop} by Arbib and Manes~\cite{arbib1975fuzzy}. Here, we additionally introduce the notion of a basis.
Intuitively, one calls a set $Y$ that is embedded into an algebraic structure $X$ a generator for the latter if every element $x$ in $X$ admits a decomposition $d(x) \in TY$ into a formal combination of elements of $Y$ that evaluates to $x$.
If the decomposition is moreover \emph{unique}, that is, $h \circ Ti$ is not only a \emph{surjection} with right-inverse $d$, but a \emph{bijection} with two-sided inverse $d$, then a generator is called a basis.
Every algebra is generated by itself using the generator $\langle X, \textnormal{id}_X, \eta_X \rangle$, but not every algebra admits a basis.
We are particularly interested in classes of set-based algebras for which \emph{every} algebra admits a \emph{size-minimal} generator, that is, no generator has a carrier of smaller size. In such a situation we will also speak of \emph{canonical} generators.

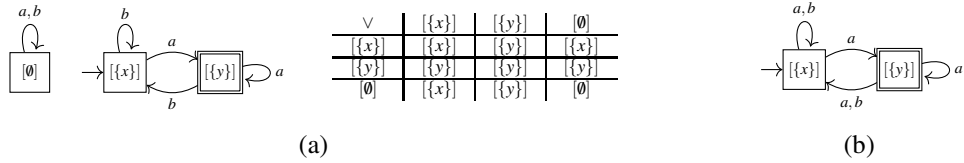
\begin{figure}[t]
	\centering
\begin{subfigure}[c]{\columnwidth}
		\centering	
		\tiny
		\begin{subfigure}[b]{.6 \columnwidth}
		\centering
		\centering
							\adjustbox{valign=m}{
			\begin{tikzpicture}[node distance=6em]
				\node[state, shape=rectangle, ] (0) {$\lbrack \emptyset \rbrack$};
				\node[state,shape=rectangle,  right of=0, initial, initial text=] (x) {$\lbrack \lbrace x \rbrace \rbrack$};
					\node[state, shape=rectangle,  right of=x, accepting] (y) {$\lbrack \lbrace y \rbrace \rbrack$};	
			    \path[->]
			(0) edge[loop above] node{$a,b$} (0)
			(x) edge[loop above] node{$b$} (x)
			(x) edge[above, bend left] node{$a$} (y)
			(y) edge[below, bend left] node{$b$} (x)
			(y) edge[loop right] node{$a$} (y)
			;
			\end{tikzpicture}
			}
			\qquad
					\adjustbox{valign=m}{
			\resizebox{0.4 \columnwidth}{!}{%
				\begin{tabular}{ c|c|c|c } 
 $\vee$ & $\lbrack \lbrace x \rbrace \rbrack$ & $ \lbrack \lbrace y \rbrace \rbrack$ &  $\lbrack \emptyset \rbrack$ \\
 \hline 
$\lbrack \lbrace x \rbrace \rbrack$ & $\lbrack \lbrace x \rbrace \rbrack$ & $\lbrack \lbrace y \rbrace \rbrack$ & $\lbrack \lbrace x \rbrace \rbrack$ \\ 
 \hline
 $\lbrack \lbrace y \rbrace \rbrack$ & $\lbrack \lbrace y \rbrace \rbrack$ & $\lbrack \lbrace y \rbrace \rbrack$ & $\lbrack \lbrace y \rbrace \rbrack$ \\
 \hline
 $\lbrack \emptyset \rbrack$ & $\lbrack \lbrace x \rbrace \rbrack$ & $\lbrack \lbrace y \rbrace \rbrack$ & $\lbrack \emptyset \rbrack$
\end{tabular}
}}
					\caption{}
		\label{overlineml}
		\end{subfigure}
	\begin{subfigure}[b]{.3 \columnwidth}
		\tiny
		\centering
							\adjustbox{valign=m}{
		\begin{tikzpicture}[node distance=6em]
			\node[state, shape=rectangle, initial, initial text=] (x) {$\lbrack \lbrace x \rbrace \rbrack$};
				\node[state, shape=rectangle, right of=x, accepting] (y) {$\lbrack \lbrace y \rbrace \rbrack$};	
		    \path[->]
		(x) edge[loop above] node{$a,b$} (x)
		(x) edge[above, bend left] node{$a$} (y)
		(y) edge[below, bend left] node{$a,b$} (x)
		(y) edge[loop right] node{$a$} (y)
		;
		\end{tikzpicture}
		}
		\caption{}
	\label{jiromaton}
	\end{subfigure}
	\end{subfigure}	
	\caption{(a) The minimal CSL-structured DFA for $\mathcal{L} = (a+b)^*a$; (b) The canonical RFSA for $\mathcal{L} = (a+b)^*a$.}
\end{figure}

\begin{example}
	\begin{itemize}
		\item 
			A tuple $\langle Y, i, d \rangle$ is a generator for a $\mathcal{P}$-algebra $L = \langle X,h \rangle \simeq \langle X, \vee^h \rangle$ iff $x = \bigvee^h_{y \in d(x)} i(y) $ for all $x \in X$. Note that if $Y \subseteq X$ is a subset, then $i(y) = y$ for all $y \in Y$.
            If $L$ satisfies the descending chain condition, which is in particular the case if $X$ is finite, then defining $i(y) = y$ and $d(x) = \lbrace y \in J(L) \mid y \leq x \rbrace$ turns the set of join-irreducibles\footnote{A non-zero element $x\in L$ is called \emph{join-irreducible} if for all $y,z \in L$ such that $x=y \vee z$ one finds $x = y $ or $x = z$.} $J(L)$ into a size-minimal generator $\langle J(L), i, d \rangle$ for $L$, cf. \ifdefined\doappendix\Cref{joinirreducstateminimal}\else \cite[Lemma~55]{arxiv}\fi.
		\item 
			A tuple $\langle Y, i, d \rangle$ is a generator for a $\mathcal{R}$-algebra $V = \langle X,h \rangle \simeq \langle X, +^h, \cdot^h \rangle$ iff $x = \sum^h_{y \in Y} d(x)(y) \cdot^h i(y)$ for all $x \in X$.
			As it is well-known that every vector space can be equipped with a basis, every $\mathcal{R}$-algebra $V$ admits a basis. One can show that a basis is size-minimal, cf. \ifdefined\doappendix\Cref{xorbasisstateminimal}\else \cite[Lemma~52]{arxiv}\fi.
	\end{itemize}
\end{example}

It is enough to find generators for the underlying algebra of a bialgebra to derive an equivalent free bialgebra.
This is because the algebraic and coalgebraic components are tightly intertwined via a distributive law.
\begin{proposition}
\label[proposition]{forgenerator-isharp-is-bialgebra-hom}
	Let $\langle X, h, k\rangle$ be a $\lambda$-bialgebra and let $\langle Y, i, d \rangle$ be a generator for the $T$-algebra $\langle X,h \rangle$.
	Then $h \circ Ti \colon \textnormal{exp}_T(\langle Y, Fd \circ k \circ i\rangle ) \rightarrow \langle X, h, k \rangle$ is a $\lambda$-bialgebra homomorphism.
\end{proposition}

Intuitively, the bialgebra $\langle X, h, k \rangle$ is a deterministic automaton with additional algebraic structure in the monad $T$ and say initial state $x \in X$, while the equivalent free bialgebra is the determinisation of the succinct automaton $Fd \circ k \circ i \colon Y \rightarrow FTY$ with side-effects in $T$ and initial state $d(x) \in TY$. 
The following result further observes that if one considers a basis for the underlying algebraic structure of a bialgebra, rather than just a generator, then the equivalent free bialgebra is in fact isomorphic to the original bialgebra.

\begin{proposition}
\label[proposition]{forbasis-isharp-is-bialgebra-iso}
	Let $\langle X, h, k\rangle$ be a $\lambda$-bialgebra and let $\langle Y, i, d \rangle$ be a basis for the $T$-algebra $\langle X,h \rangle$.
	Then $h \circ Ti \colon \textnormal{exp}_T(\langle Y, Fd \circ k \circ i\rangle ) \rightarrow \langle X, h, k \rangle$ is a $\lambda$-bialgebra isomorphism.
\end{proposition}

We conclude this section by illustrating how \Cref{forgenerator-isharp-is-bialgebra-hom} can be used to construct the canonical RFSA \cite{DenisLT02}, the canonical nominal RFSA \cite{moerman2019residual}, and the minimal xor automaton \cite{VuilleminG210} for a regular language $\mathcal{L}$ over some alphabet $A$. All examples follow three analogous steps:
\begin{enumerate}
	\item We construct the minimal\footnote{Minimal in the sense that every state is reachable by an element of $A^*$ and no two different states observe the same language.} pointed coalgebra $M_{\mathcal{L}}$ for the (nominal) set endofunctor $F = 2 \times (-)^{A}$ accepting $\mathcal{L}$. For the case $A = \lbrace a, b \rbrace$ and $\mathcal{L} = (a+b)^*a$, the coalgebra $M_{\mathcal{L}}$ is depicted in \Cref{m(l)}.
	\item We equip the former with additional algebraic structure in a monad $T$ (which is related to $F$ via a canonically induced distributive law $\lambda$) by generating the $\lambda$-bialgebra $\textnormal{free}_T(M_{\mathcal{L}})$. By identifying semantically equivalent states we consequently derive the minimal\footnote{Minimal in the sense that every state is reachable by an element of $T(A^*)$ and no two different states observe the same language.} (pointed) $\lambda$-bialgebra $\langle X, h, k \rangle$  for $\mathcal{L}$.
	\item We identify canonical generators $\langle Y, i, d \rangle$ for $\langle X, h \rangle$ and use \Cref{forgenerator-isharp-is-bialgebra-hom} to derive an equivalent succinct automaton $\langle Y, Fd \circ k \circ i \rangle$ with side-effects in $T$.
\end{enumerate}

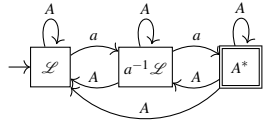
\begin{figure*}
\tiny
\centering
	\begin{tikzpicture}[node distance=6em]
			\node[state, shape=rectangle, initial, initial text=] (L) {$\mathcal{L}$};
				\node[state, shape=rectangle, right of=L] (aL) {$a^{-1}\mathcal{L}$};	
\node[state, accepting, shape=rectangle, right of=aL] (A) {$A^*$};	
		    \path[->]
		(L) edge[loop above] node{$A$} (L)
		(L) edge[above, bend left] node{$a$} (aL)
		(aL) edge[above, bend left] node{$a$} (A)
		(aL) edge[loop above] node{$A$} (aL)
		(aL) edge[above, bend left] node{$A$} (L)
		(A) edge[loop above] node{$A$} (A)
		(A) edge[above, bend left=45] node{$A$} (L)
		(A) edge[above, bend left] node{$A$} (aL)
		;
		\end{tikzpicture}
		\caption{The orbit-finite representation of the canonical nominal RFSA for $\mathcal{L} =  \lbrace v a w a u \mid v, w, u \in A^*, a \in A \rbrace$.}
				\label{fig:canonialnominalrfsa}
\end{figure*}

\begin{example}\deftag{The canonical RFSA}
\label[example]{canonicalrfsaexample}
    Using the $\mathcal{P}$-algebra structure $h^{\mathcal{P}}: \mathcal{P}2 \rightarrow 2$ with $h^{\mathcal{P}}(\varphi) = \varphi(1)$, we derive a canonical distributive law $\lambda^{\mathcal{P}}$ between $F$ and the powerset monad $\mathcal{P}$. The minimal pointed $\lambda^{\mathcal{P}}$-bialgebra for $\mathcal{L} = (a+b)^*a$ with its underlying CSL structure is depicted in \Cref{overlineml}; the construction can be verified with the help of \ifdefined\doappendix\Cref{freepowersetbialgebrastructure}\else \cite[Lemma~47]{arxiv}\fi. The partially ordered state space $L = \lbrace \lbrack \emptyset \rbrack \leq \lbrack \lbrace x \rbrace \rbrack \leq \lbrack \lbrace y \rbrace \rbrack \rbrace$ is necessarily finite, thus satisfies the descending chain condition, which turns the set of join-irreducibles into a size-minimal generator $\langle J(L), i, d\rangle$ with $i(y) = y$ and $d(x) = \lbrace y \in J(L) \mid y \leq x \rbrace$, cf. \ifdefined\doappendix\Cref{joinirreducstateminimal}\else \cite[Lemma~55]{arxiv}\fi. In this case, the join-irreducibles are given by all non-zero states. The $\mathcal{P}$-succinct automaton consequently induced by \Cref{forgenerator-isharp-is-bialgebra-hom} is depicted in \Cref{jiromaton}; it can be recognised as the canonical RFSA, cf. e.g. \cite{MyersAMU15}.
 \end{example}

\begin{example}\deftag{The canonical nominal RFSA}
\label[example]{nominalexample}
    It is not hard to see that $F$ extends to a functor on the category of nominal sets; the usual strength function is equivariant \ifdefined\doappendix(\Cref{equivariantstrength})\else\cite[Lemma~46]{arxiv}\fi; and $h^{\mathcal{P}_{\textnormal{n}}}: \mathcal{P}_{\textnormal{n}}2 \rightarrow 2$ with $h^{\mathcal{P}_{\textnormal{n}}}(\varphi) = \varphi(1)$ defines a $\mathcal{P}_{\textnormal{n}}$-algebra, which induces a canonical distributive law $\lambda^{\mathcal{P}_{\textnormal{n}}}$ between $F$ and the nominal powerset monad $\mathcal{P}_{\textnormal{n}}$. As in \cite{moerman2019residual}, let $\mathcal{L} = \lbrace v a w a u \mid v, w, u \in A^*, a \in A \rbrace$, then $a^{-n}\mathcal{L} = a^{-2}\mathcal{L} = A^*$ for $n \geq 2$, and $v^{-1}\mathcal{L} = \cup_{a \in A} a^{-\vert v \vert_a} \mathcal{L}$, where $\vert v \vert_a$ denotes the number of $a$'s that occur in $v$. In consequence, the nominal CSL underlying the minimal pointed $\lambda^{\mathcal{P}_{\textnormal{n}}}$-bialgebra is generated by the orbit-finite nominal set of join-irreducibles $\lbrace \mathcal{L} \rbrace \cup \lbrace a^{-1} \mathcal{L} \mid a \in A \rbrace \cup \lbrace A^* \rbrace$, which is equipped with the obvious $\textnormal{Perm}(A)$-action and satisfies the inclusion $\mathcal{L} \subseteq a^{-1} \mathcal{L} \subseteq A^*$. The orbit-finite representation of the $\mathcal{P}_{\textnormal{n}}$-succinct automaton consequently induced by \Cref{forgenerator-isharp-is-bialgebra-hom} is depicted in \Cref{fig:canonialnominalrfsa}. 
\end{example}

\begin{example}\deftag{The minimal xor automaton}
\label[example]{minimalxorexmaple}
    The $\mathcal{R}$-algebra structure $h^{\mathcal{R}}: \mathcal{R}2 \rightarrow 2$ with $h^{\mathcal{R}}(\varphi) = \varphi(1)$ induces a canonical distributive law $\lambda^{\mathcal{R}}$ between $F$ and the free vector space monad $\mathcal{R}$ over the two element field. The minimal pointed $\lambda^{\mathcal{R}}$-bialgebra accepting $\mathcal{L} = (a+b)^*a$ is depicted in \Cref{fm(l)xor} and coincides with the bialgebra freely generated by the $F$-coalgebra in \Cref{m(l)}. The construction can be verified using \ifdefined\doappendix\Cref{freexorbialgebrastructure}\else\cite[Lemma~50]{arxiv}\fi.
    The underlying vector space structure necessarily has a basis; we choose the size-minimal generator  $\langle Y, i, d \rangle$ with $Y = \lbrace \lbrace x \rbrace, \lbrace x, y \rbrace \rbrace$, $i(y) = y$, and $d(\emptyset) = \emptyset$, $d(\lbrace x \rbrace) = \lbrace \lbrace x \rbrace \rbrace$, $d(\lbrace y \rbrace) = \lbrace \lbrace x \rbrace$, $\lbrace x, y \rbrace \rbrace$, $d(\lbrace x, y \rbrace) = \lbrace \lbrace x, y \rbrace \rbrace$, which is sufficient by \ifdefined\doappendix\Cref{xorbasisstateminimal}\else\cite[Lemma~52]{arxiv}\fi. The $\mathcal{R}$-succinct automaton induced by \Cref{forgenerator-isharp-is-bialgebra-hom} is depicted in \Cref{xorautomaton}; it can be recognised as the minimal xor automaton, cf. e.g. \cite{MyersAMU15}.  
\end{example}

\section{Changing the type of succinct automata}

\label{succinctbialgebra2}

This section contains a generalisation of the approach in \Cref{succinctbialgebra1}. The extension is based on the observation that in the last section we implicitly considered \textit{two} types of monads: (i) a monad $S$ that describes the additional algebraic structure of a given deterministic automaton; and (ii) a monad $T$ that captures the side-effects of the succinct automaton that is obtained by the generator-based translation.
In \Cref{forgenerator-isharp-is-bialgebra-hom}, the main result of the last section, the monads coincided, but to recover for instance the \'atomaton \cite{BrzozowskiT14} we will have to extend \Cref{forgenerator-isharp-is-bialgebra-hom} to a situation where $S$ and $T$ can differ. 

\subsection{Relating distributive laws}

We now introduce the main technical ingredient of our extension:  \textit{distributive law homomorphisms}.
As before, we present the theory on the level of arbitrary bialgebras, even though we will later focus on the case where the coalgebraic dynamics are those of deterministic automata. Distributive law homomorphisms will allow us to shift a bialgebra over a monad $S$ to an equivalent bialgebra over a monad $T$, for which we can then find, analogous to \Cref{succinctbialgebra1}, an equivalent succinct representation. The notion we use is an instance of a much more general definition that allows to relate distributive laws on two different categories. We restrict to the case where both distributive laws are given over the same behavioural endofunctor $F$.

\begin{definition}\deftag{Distributive law homomorphism \cite{watanabe2002well, power2002combining}}
Let $\lambda^{S}: SF \rightarrow FS$ and $\lambda^{T}: TF \rightarrow FT$ be distributive laws between monads $S$ and $T$ and an endofunctor $F$, respectively.
	A \textnormal{distributive law homomorphism} $\alpha: \lambda^{S} \rightarrow \lambda^{T}$ consists of a natural transformation $\alpha: T \Rightarrow S$ satisfying $\mu^S \circ \alpha_S \circ T \alpha = \alpha \circ \mu^T$, $\alpha \circ \eta^T = \eta^S$ and $\lambda^S \circ \alpha_F = F\alpha \circ \lambda^T$.
\end{definition}

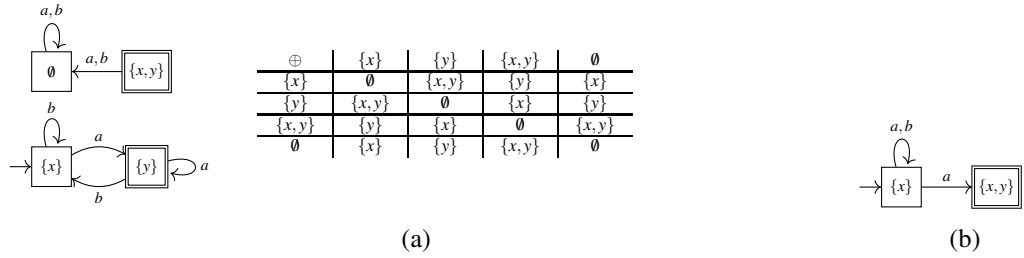
\begin{figure}[t]
\centering
\begin{subfigure}[c]{\columnwidth}
\centering
\begin{subfigure}[b]{.7 \columnwidth}
\tiny
					\adjustbox{valign=m}{
\begin{tikzpicture}[node distance=6em]
	\node[state, shape=rectangle] (0) {$\emptyset$};
		\node[state, shape=rectangle, right of=0, accepting] (xory) {$\lbrace x,y \rbrace$};
	\node[state, shape=rectangle, below of=0, initial, initial text=] (x) {$\lbrace x \rbrace$};
	\node[state, shape=rectangle, right of=x, accepting] (y) {$\lbrace y \rbrace$};	
	    \path[->]
	(0) edge[loop above] node{$a,b$} (0)
	(xory) edge[above] node{$a,b$} (0)
	(x) edge[loop above] node{$b$} (x)
	(x) edge[above, bend left] node{$a$} (y)
	(y) edge[below, bend left] node{$b$} (x)
	(y) edge[loop right] node{$a$} (y)
	;
	\end{tikzpicture}
	}
	\qquad
						\adjustbox{valign=m}{
				\resizebox{0.45 \columnwidth}{!}{%
		\begin{tabular}[]{ c|c|c|c|c } 
 $\oplus$ & $ \lbrace x \rbrace $ & $ \lbrace y \rbrace$ & $\lbrace x, y \rbrace$ &   $\emptyset$  \\
 \hline
 $ \lbrace x \rbrace $ & $\emptyset$ & $\lbrace x, y \rbrace$ & $\lbrace y \rbrace$ & $\lbrace x \rbrace$ \\
 \hline
 $ \lbrace y \rbrace$ & $\lbrace x, y \rbrace$ & $\emptyset$ & $\lbrace x \rbrace$ & $\lbrace y \rbrace$\\
\hline
 $\lbrace x, y \rbrace$ & $\lbrace y \rbrace$ & $\lbrace x \rbrace$ & $\emptyset$ & $\lbrace x, y \rbrace$\\
\hline
 $\emptyset$ & $ \lbrace x \rbrace $ & $ \lbrace y \rbrace$ & $\lbrace x, y \rbrace$ &   $\emptyset$
\end{tabular}
}}
	\caption{}
\label{fm(l)xor}
\end{subfigure}
\begin{subfigure}[b]{.2 \columnwidth}
\tiny
					\adjustbox{valign=m}{
\begin{tikzpicture}[node distance=6em]
		\node[state, shape=rectangle, initial, initial text=] (x) {$\lbrace x \rbrace$};
		\node[state, shape=rectangle, right of=x, accepting] (xory) {$\lbrace x,y \rbrace$};
	    \path[->]
	(x) edge[loop above] node{$a,b$} (x)
	(x) edge[above] node{$a$} (xory)
	;
	\end{tikzpicture}
	}
	\caption{}
\label{xorautomaton}
\end{subfigure}
\end{subfigure}
\caption{(a) The minimal $\mathbb{Z}_2$-Vect structured DFA for $\mathcal{L} = (a+b)^*a$ (freely-generated by the DFA in \Cref{m(l)}); (b) Up to the choice of a basis, the minimal xor automaton for $\mathcal{L} = (a+b)^*a$.}
\end{figure}

The above definition is such that $\alpha$ induces a functor between the categories of $\lambda^S$- and $\lambda^T$-bialgebras.

\begin{lemma}\tagcite{klin2015presenting, bonsangue2013presenting}
\label[lemma]{inducedbialgebra}
			Let $\alpha: \lambda^S \rightarrow \lambda^T$ be a distributive law homomorphism. Then $\alpha \langle X, h, k\rangle := \langle X, h \circ \alpha_X, k \rangle$ and  $\alpha(f) := f$ defines a functor $\alpha: \textnormal{Bialg}(\lambda^S) \rightarrow \textnormal{Bialg}(\lambda^T)$.
\end{lemma}

The next result is a straightforward consequence of \Cref{forgenerator-isharp-is-bialgebra-hom}, and may be strengthened to an isomorphism in case one is given a basis instead of a generator, analogous to \Cref{forbasis-isharp-is-bialgebra-iso}. It can be seen as a road map to the approach we propose in this section.

\begin{corollary}
\label[corollary]{generatorbialgebrahom}
	Let $\alpha: \lambda^S \rightarrow \lambda^T$ be a homomorphism between distributive laws and $\langle X,h,k \rangle$ a $\lambda^S$-bialgebra. If $\langle Y, i, d \rangle$ is a generator for the $T$-algebra $\langle X, h \circ \alpha_X \rangle$, then $
	(h \circ \alpha_X) \circ Ti: \textnormal{exp}_T(\langle  Y, Fd \circ k \circ i \rangle) \rightarrow \langle X, h \circ \alpha_X, k \rangle
	$ is a $\lambda^T$-bialgebra homomorphism.
\end{corollary}

\subsection{Deriving distributive law  relations}

We now turn to the procedure of deriving a distributive law homomorphism. In practice, coming up with a natural transformation and proving that it lifts to a distributive law homomorphism can be quite cumbersome. 

 Fortunately, for certain cases, there is a way to simplify things significantly. For instance, as the next result shows, if, as in \eqref{induceddistrlaweq}, the involved distributive laws are induced by algebra structures $h^{S}$ and $h^{T}$ for an output set $B$, respectively, then one of the conditions is implied by a less convoluted constraint.

\begin{lemma}
\label[lemma]{distributivelawaxiomeasier}
	Let $\alpha: T \Rightarrow S$ be a natural transformation satisfying $h^{S} \circ \alpha_B = h^{T}$, then $\lambda^{S} \circ \alpha_F = F \alpha \circ \lambda^{T}$.
\end{lemma}

The next result shows that for the neighbourhood monad there exists a family of \textit{canonical} choices of distributive law homomorphisms parametrised by Eilenberg-Moore algebra structures on the output set $B = 2$. While it is well-known that such algebras induce a monad morphism, for instance in the coalgebraic modal logic community \cite{klin2004coalgebraic, schroder2008expressivity, hansen2014strong}, its commutativity with canonical distributive laws has not been observed before. Moreover, we provide a new formalisation in terms of the strength function, which allows the result to be lifted to strong monads and arbitrary output objects on other categories than the one of sets and functions.

\begin{proposition}
\label[proposition]{algebrainduceddistributivellawhom}
	Any algebra $h: T2 \rightarrow 2$ over a set monad $T$ induces a homomorphism $\alpha^{h}: \lambda^{\mathcal{H}} \rightarrow \lambda^{h}$ between distributive laws by $\alpha^{h}_X := h^{2^X} \circ \textnormal{st} \circ T(\eta^{\mathcal{H}}_X)$.\end{proposition}

The rest of the section is concerned with using \Cref{algebrainduceddistributivellawhom} and \Cref{generatorbialgebrahom} to derive canonical acceptors based on induced distributive law homomorphisms.

\begin{figure}[t]
\centering
\begin{subfigure}[c]{\columnwidth} 
\centering
\begin{subfigure}[b]{0.5 \columnwidth}
\center
\tiny
\begin{tikzpicture}[node distance= 2em, ]
	\node[state, initial, shape=rectangle, initial text=] (1) {$\lbrack \lbrace \lbrace x \rbrace, \lbrace x, y \rbrace \rbrace \rbrack$};

	\node[state, shape=rectangle, right = of 1] (3) {$\lbrack \emptyset \rbrack$};

	\node[state, shape=rectangle, below = of 1, accepting] (5) {$\lbrack \lbrace \lbrace x \rbrace,  \lbrace y \rbrace, \lbrace x, y \rbrace \rbrace \rbrack$};

	\node[state, shape=rectangle, right = of 5, accepting] (15) {$\lbrack \lbrace \lbrace x \rbrace, \lbrace y \rbrace, \lbrace x, y \rbrace, \emptyset \rbrace \rbrack$};
		    \path[->]
	(3) edge[loop above] node{$a,b$} (3)
	(1) edge[loop above] node{$b$} (1)
	(1) edge[right, bend left] node{$a$} (5)
	(5) edge[loop below] node{$a$} (5)
	(5) edge[left, bend left] node{$b$} (1)
	(15) edge[loop below] node{$a,b$} (15)
;
	\end{tikzpicture} \\
		\resizebox{0.3 \columnwidth}{!}{%
		\begin{tabular}[]{ c|c|c|c|c} 
	$\vee$ & $1$ & $2$ & $3$ & $4$ \\
	\hline
	$1$ & $1$ & $1$ & $3$ & $4$ \\
	\hline
	$2$ & $1$ & $2$ & $3$ & $4$ \\
	\hline
	$3$ & $3$ & $3$ & $3$ & $4$ \\
	\hline
	$4$ & $4$ & $4$ & $4$ & $4$\\
	\end{tabular}
	}
			\resizebox{0.3 \columnwidth}{!}{%
			\begin{tabular}[]{ c|c|c|c|c} 
	$\wedge$ & $1$ & $2$ & $3$ & $4$ \\
	\hline
	$1$ & $1$ & $2$ & $1$ & $1$ \\
	\hline
	$2$ & $2$ & $2$ & $2$ & $2$ \\
	\hline
	$3$ & $1$ & $2$ & $3$ & $3$ \\
	\hline
	$4$ & $1$ & $2$ & $3$ & $4$
	\end{tabular}
	}

\caption{}
\label{overlinem(l)distro}
\end{subfigure} 
\begin{subfigure}[b]{0.4 \columnwidth}
\centering
\tiny
\begin{tikzpicture}[node distance=2em]
	\node[state, initial, shape=rectangle, initial text=] (1) {$\lbrack \lbrace \lbrace x \rbrace, \lbrace x, y \rbrace \rbrace \rbrack$};

	\node[state, shape=rectangle,below= of  1, accepting] (5) {$\lbrack \lbrace \lbrace x \rbrace,  \lbrace y \rbrace, \lbrace x, y \rbrace \rbrace \rbrack$};

	\node[state, shape=rectangle, right= of  5, accepting] (15) {$\lbrack \lbrace \lbrace x \rbrace, \lbrace y \rbrace, \lbrace x, y \rbrace, \emptyset \rbrace \rbrack$};
		    \path[->]
	(1) edge[loop above] node{$a,b$} (1)
	(1) edge[right, bend left] node{$a$} (5)
	(5) edge[loop below] node{$a$} (5)
	(5) edge[left, bend left] node{$a,b$} (1)
	(15) edge[loop below] node{$a,b$} (15)
	(15) edge[bend right, right] node{$a,b$} (1)
	(15) edge[below] node{$a,b$} (5)
;
	\end{tikzpicture}
\caption{}
\label{distromaton}
\end{subfigure} 
\end{subfigure}
\caption{(a) The minimal CDL-structured DFA for $\mathcal{L} = (a+b)^*a$, where
$1 \equiv \lbrack \lbrace \lbrace x \rbrace, \lbrace x, y \rbrace \rbrace \rbrack$, $2 \equiv \lbrack \emptyset \rbrack$, $3 \equiv \lbrack \lbrace \lbrace x \rbrace, \lbrace y \rbrace, \lbrace x, y \rbrace \rbrace \rbrack$, $4 \equiv \lbrack \lbrace \lbrace x \rbrace, \lbrace y \rbrace, \lbrace x, y \rbrace, \emptyset \rbrace \rbrack$; (b) The distromaton for $\mathcal{L} = (a+b)^*a$.}
\end{figure}
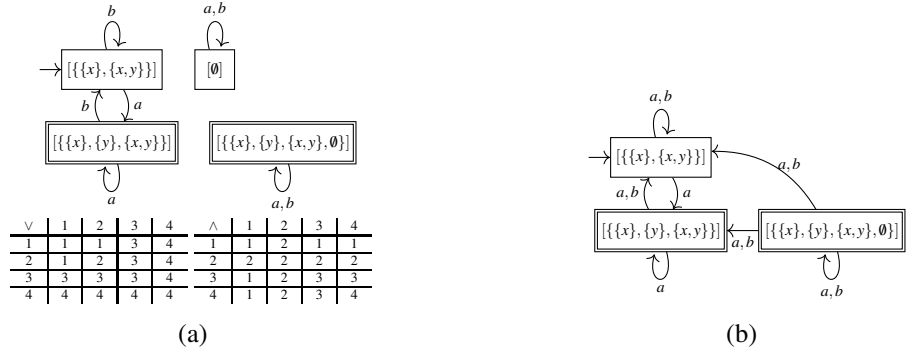

\subsection{Example: The \'atomaton}

\label{atomatonexample}
We will now justify the previous informal construction of the \'atomaton. 
As hinted before, the \'atomaton can be recovered by relating the neighbourhood monad $\mathcal{H}$---whose algebras are complete \emph{atomic} Boolean algebras (CABAs)---and the powerset monad $\mathcal{P}$. Formally, as a consequence of \Cref{algebrainduceddistributivellawhom} we obtain the following. 

\begin{corollary}
\label[corollary]{alphapowersetneighbourhooddistrlaw}
	Let $\alpha_X: \mathcal{P}X \rightarrow \mathcal{H}X$ satisfy $\alpha_X(\varphi)(\psi) = \bigvee_{x \in X} \varphi(x) \wedge \psi(x)$, then $\alpha$ constitutes a distributive law homomorphism $\alpha: \lambda^{\mathcal{H}} \rightarrow \lambda^{\mathcal{P}}$.
\end{corollary}

The next statement follows from a well-known Stone-type duality \cite{taylor2002subspaces} representation theorem for CABAs.

\begin{lemma}
\label[lemma]{basisshiftecabapowerset}
	Let $\alpha_X: \mathcal{P}X \rightarrow \mathcal{H}X$ satisfy $\alpha_X(\varphi)(\psi) = \bigvee_{x \in X} \varphi(x) \wedge \psi(x)$. If $B = \langle X, h \rangle$ is a $\mathcal{H}$-algebra, then $\langle \textnormal{At}(B), i, d \rangle$ with $i(a) = a$ and $d(x) = \lbrace a \in \textnormal{At}(B) \mid a \leq x \rbrace$ is a basis for the $\mathcal{P}$-algebra $\langle X, h \circ \alpha_X \rangle$.
\end{lemma}

The \'atomaton for the regular language $\mathcal{L} = (a+b)^*a$, for example, can now be obtained as follows.
First, we construct the minimal pointed $\lambda^{\mathcal{H}}$-bialgebra accepting $\mathcal{L}$, which is depicted in \Cref{overlinem(l)atom} together with its underlying CABA structure $B$. The construction can be verified with the help of \ifdefined\doappendix\Cref{freeneighbourhoodbialgebrastructure}\else\cite[Lemma~48]{arxiv}\fi. Using the distributive law homomorphism $\alpha$ of \Cref{alphapowersetneighbourhooddistrlaw}, it can be translated into an equivalent pointed $\lambda^{\mathcal{P}}$-bialgebra with underlying CSL-structure $\alpha(B)$. By \Cref{basisshiftecabapowerset} the atoms $\textnormal{At}(B)$ of $B$ form a basis for $\alpha(B)$. In this case the atoms are given by $\lbrack \lbrace \lbrace x \rbrace, \lbrace x, y \rbrace \rbrace \rbrack, \lbrack \lbrace \lbrace y \rbrace \rbrace \rbrack$ and $\lbrack \lbrace \emptyset \rbrace \rbrack$. The $\mathcal{P}$-succinct automaton consequently induced by \Cref{generatorbialgebrahom} is depicted in \Cref{atomaton}; it can be recognised as the \'atomaton, cf. e.g. \cite{MyersAMU15}.

\subsection{Example: The distromaton}

\label{distromatonexample}

We shall now use our framework to recover another canonical non-deterministic acceptor: the \textit{distromaton} \cite{MyersAMU15}. As the name suggests, it can be constructed by relating the monotone neighbourhood monad $\mathcal{A}$---whose algebras are completely \textit{distributive} lattices---and the powerset monad $\mathcal{P}$. Formally, the relationship can be established by the same natural transformation we used for the \'atomaton. 

\begin{corollary}
\label[corollary]{neighbourhoodpowersetmorphism}
	Let $\alpha_X: \mathcal{P}X \rightarrow \mathcal{A}X$ satisfy $\alpha_X(\varphi)(\psi) = \bigvee_{x \in X} \varphi(x) \wedge \psi(x)$, then $\alpha$ constitutes a distributive law homomorphism $\alpha: \lambda^{\mathcal{A}} \rightarrow \lambda^{\mathcal{P}}$.
\end{corollary}

The distromaton for the regular language $\mathcal{L} = (a+b)^*a$, for example, can now be obtained as follows.
First, we construct the minimal  pointed $\lambda^{\mathcal{A}}$-bialgebra for $\mathcal{L}$, depicted in \Cref{overlinem(l)distro} with its underlying CDL structure $h$. The construction can be verified with the help of \ifdefined\doappendix\Cref{freealternatingbialgebrastructure}\else\cite[Lemma~49]{arxiv}\fi. Using the distributive law homomorphism $\alpha$ in \Cref{neighbourhoodpowersetmorphism}, it can be translated into an equivalent pointed $\lambda^{\mathcal{P}}$-bialgebra with underlying CSL structure $L = h \circ \alpha_X$. Its partially ordered state space $  \lbrack \emptyset \rbrack \leq \lbrack \lbrace \lbrace x \rbrace, \lbrace x, y \rbrace \rbrace \rbrack \leq \lbrack \lbrace \lbrace x \rbrace, \lbrace y \rbrace, \lbrace x, y \rbrace \rbrace \rbrack \leq \lbrack \lbrace \lbrace x \rbrace, \lbrace y \rbrace, \lbrace x, y \rbrace, \emptyset \rbrace \rbrack  $ is necessarily finite, which turns the set of join-irreducibles into a size-minimal generator $\langle J(L), i, d \rangle$ for $L$, where $i(y) = y$ and $d(x) = \lbrace y \in J(L) \mid y \leq x \rbrace$. In this case, the join-irreducibles are given by all non-zero states.   
	The $\mathcal{P}$-succinct automaton consequently induced by \Cref{generatorbialgebrahom} is depicted in \Cref{distromaton} and can be recognised as the distromaton, cf. \cite{MyersAMU15}.

\subsection{Example: The minimal xor-CABA automaton}

\label{minimalxorcabaexample}

We conclude this section by relating the neighbourhood monad $\mathcal{H}$ with the free vector space monad $\mathcal{R}$ over the unique two element field $\mathbb{Z}_2$. In particular, we derive a new canonical succinct acceptor for regular languages, which we call the \emph{minimal xor-CABA automaton}. 

Intuitively, the next result says that every CABA can be equipped with a symmetric difference like operation that turns it into a vector space over the two element field. 

\begin{corollary}
\label[corollary]{alphaxorneighbourhooddistrlaw}
	Let $\alpha_X: \mathcal{R}X \rightarrow \mathcal{H}X$ satisfy $\alpha_X(\varphi)(\psi) = \bigoplus_{x \in X} \varphi(x) \cdot \psi(x)$, then $\alpha$ constitutes a distributive law homomorphism $\alpha: \lambda^{\mathcal{H}} \rightarrow \lambda^{\mathcal{R}}$.
\end{corollary}

Since every vector space admits a basis, above result leads to the definition of a new acceptor of regular languages.
Let $\alpha$ denote the homomorphism in \Cref{alphaxorneighbourhooddistrlaw} and $F$ the endofunctor given by $FX = 2 \times X^{A}$.

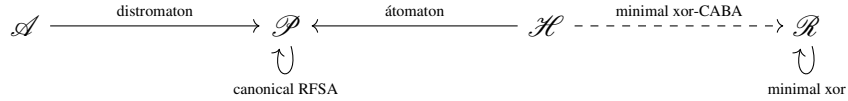
\begin{figure*}[t]
\centering
	\begin{tikzpicture}[node distance=9em]
		\node[] (H) {$\mathcal{H}$};
		\node[left of = H] (P) {$\mathcal{P}$};
		\node[right of = H] (X) {$\mathcal{R}$};
		\node[left of = P] (A) {$\mathcal{A}$};
		    \path[->]
	(H) edge[above] node{\tiny \'atomaton} (P)
	(H) edge[above, dashed] node{\tiny
 minimal xor-CABA} (X)
	(P) edge[loop below] node{\tiny
 canonical RFSA} (P)
	(X) edge[loop below] node{\tiny
minimal xor} (X)
	(A) edge[above] node{\tiny
 distromaton} (P)
;
	\end{tikzpicture}
	\caption{The minimal xor-CABA automaton is to the minimal xor automaton what the \'atomaton is to the canonical RFSA.}
	\label{minimalxorcabadiagram}
\end{figure*}

\begin{definition}\deftag{Minimal xor-CABA automaton}
	Let $\langle X, h, k \rangle$ be the minimal $x$-pointed $\lambda^{\mathcal{H}}$-bialgebra accepting a regular language $\mathcal{L} \subseteq A^*$, and $B = \langle Y, i, d \rangle$ a basis for the $\mathcal{R}$-algebra $\langle X, h \circ \alpha_X \rangle$. The \emph{minimal xor-CABA automaton} for $\mathcal{L}$ with respect to $B$ is the $d(x)$-pointed $\mathbb{Z}_2$-weighted automaton $Fd \circ k \circ i$.
	\end{definition}
	
	In \Cref{minimalxorcabadiagram} it is indicated how the canonical acceptors of this paper, including the minimal xor-CABA automaton, are based on relations between pairs of monads.
		
	For the regular language $\mathcal{L} = (a+b)^*a$ above definition instantiates as follows. First, as for the \'atomaton, we construct the minimal pointed $\lambda^{\mathcal{H}}$-bialgebra $\langle X, h, k\rangle$ for $\mathcal{L}$; it is depicted in \Cref{overlinem(l)atom}. As one easily verifies, the $\mathbb{Z}_2$-vector space $\langle X, h \circ \alpha_X \rangle$ is induced by the symmetric difference operation $\oplus$ on subsets. Using the notation in \Cref{overlinem(l)atom}, we choose the basis $\langle Y, i, d \rangle$ with $Y = \lbrace  4,6,7,8 \rbrace$; $i(y) = y$; and $d(1) = 7 \oplus 8$, $d(2) = \emptyset$, $d(3) = 6 \oplus 7$, $d(4) = 4$, $d(5) = 6 \oplus 7 \oplus 8$, $d(6) = 6$, $d(7) = 7$, $d(8) = 8$. The induced $d(1) = 7 \oplus 8$-pointed $\mathcal{R}$-succinct automaton accepting $\mathcal{L}$, i.e. the minimal xor-CABA automaton, is depicted in \Cref{minimalxorcaba}.

\section{Minimality}\label{minimality}

\newcommand{\img}{\textnormal{im}}
\newcommand{\gen}{\textnormal{gen}}
\newcommand{\obs}{\textnormal{obs}}
\newcommand{\ext}{\textnormal{ext}}
\newcommand{\expa}{\textnormal{exp}}

 In this section we restrict ourselves to the category of (nominal) sets. We show that every language satisfying a suitable property parametric in monads $S$ and $T$ admits a size-minimal succinct automaton of type $T$ accepting it. As a main result we obtain \Cref{minimalitytheorem}, which is a generalisation of parts of \cite[Theorem 4.8]{MyersAMU15}. In \Cref{minimalityimplications} we instantiate the former to subsume known minimality results for canonical automata, to prove the xor-CABA automaton minimal, and to establish a size-comparison between different acceptors.

Given a distributive law homomorphism $\alpha: \lambda^S \rightarrow \lambda^T$, let $\textnormal{ext}: \textnormal{Coalg}(FT) \rightarrow \textnormal{Coalg}(FS)$ be the functor given by $\textnormal{ext}(\langle X, k \rangle) = \langle X, F\alpha_X \circ k \rangle$ and $\textnormal{ext}(f) = f$. Moreover, let $\textnormal{exp}_U: \textnormal{Coalg}(FU) \rightarrow \textnormal{Bialg}(\lambda^U)$ for $U \in \lbrace S, T \rbrace$ denote the functor introduced in \Cref{expfunctor}.

\begin{proposition}
\label[proposition]{alphaunderlies}
Let $\alpha: \lambda^{S} \rightarrow \lambda^T$ be a distributive law homomorphism. Then $\alpha_X: TX \rightarrow SX$ underlies a natural transformation $\alpha: \textnormal{exp}_T \Rightarrow \alpha \circ \textnormal{exp}_S \circ \textnormal{ext}$ between functors of type $\textnormal{Coalg}(FT) \rightarrow \textnormal{Bialg}(\lambda^T)$.
\end{proposition}

In the above situation a $T$-succinct automaton admits \emph{two} semantics, induced by lifting the former either to a bialgebra over $\lambda^S$ or $\lambda^T$.
The next definition introduces a notion of \emph{closedness} that captures those cases in which the image of both semantics coincides.

\begin{definition}\deftag{$\alpha$-closed succinct automaton}
\label[definition]{closedsuccinctdef}
Let $\alpha: \lambda^{S} \rightarrow \lambda^T$ be a distributive law homomorphism.
	We say that a $T$-succinct automaton $\mathcal{X}$ is \emph{$\alpha$-closed} if the unique diagonal below is an isomorphism:
	\[
	\begin{tikzcd}
		\expa_T(\mathcal{X}) \arrow[twoheadrightarrow]{r}{\textnormal{obs}} \arrow{d}[left]{\textnormal{obs} \circ \alpha_X} & \img(\obs_{\expa_T(\mathcal{X})}) \arrow[dashed]{dl}{} \arrow{d}{} \\
		\img(\obs_{\alpha(\expa_S(\ext(\mathcal{X})))}) \arrow[hookrightarrow]{r}{} & \Omega
	\end{tikzcd}.
	\]
	\end{definition}

Next we show that succinct automata obtained from certain generators are $\alpha$-closed.

\begin{lemma}
\label[lemma]{generatorclosed}
Let $\alpha: \lambda^{S} \rightarrow \lambda^T$ be a distributive law homomorphism and $\langle X, h, k \rangle$ a $\lambda^S$-bialgebra.
	If $\langle Y, i, d \rangle$ is a generator for $\langle X, h \circ \alpha_X \rangle$, then $\langle Y,  Fd \circ k \circ i\rangle$ is $\alpha$-closed.
\end{lemma}

We are now able to state our main result, which is a generalisation of parts of \cite[Theorem 4.8]{MyersAMU15}.

\begin{theorem}[Minimal succinct automata]
\label[theorem]{minimalitytheorem}
	Given a language $\mathcal{L} \in \Omega$ such that there exists a minimal pointed $\lambda^S$-bialgebra $\mathbb{M}$ accepting $\mathcal{L}$ and the underlying algebra of $\alpha(\mathbb{M})$ admits a size-minimal generator, there exists a pointed $\alpha$-closed $T$-succinct automaton $\mathcal{X}$ accepting $\mathcal{L}$ such that: \begin{itemize}
		\item for any pointed $\alpha$-closed $T$-succinct automaton $\mathcal{Y}$ accepting $\mathcal{L}$ we have that $\img(\obs_{\expa_T(\mathcal{X})}) \subseteq \img(\obs_{\expa_T(\mathcal{Y})})$;
		\item if $\img(\obs_{\expa_T(\mathcal{X})}) = \img(\obs_{\expa_T(\mathcal{Y})})$, then $\vert X \vert \leq \vert Y \vert$, where $X$ and $Y$ are the carriers of $\mathcal{X}$ and $\mathcal{Y}$, respectively.
	\end{itemize} 
 \end{theorem}

For a $T$-succinct automaton $\mathcal{X}$ let us write $\textnormal{obs}^{\dag}_{\mathcal{X}} := \textnormal{obs}_{\textnormal{exp}_T(\mathcal{X})} \circ \eta^T_X: X \rightarrow \Omega$ for a generalisation of the semantics of non-deterministic automata. The next result provides an equivalent characterisation of $\alpha$-closedness in terms of $\textnormal{obs}^{\dag}$ that will be particularly useful in \Cref{minimalityimplications}. 

\begin{lemma}
\label[lemma]{imgobssuccinctsem}
	Let $\alpha: \lambda^{S} \rightarrow \lambda^T$ be a distributive law homomorphism. For any $T$-succinct automaton $\mathcal{X}$ it holds that $\img(\obs_{\exp_T(\mathcal{X})}) = \img(h \circ \alpha_{\Omega} \circ T(\obs^{\dag}_{\mathcal{X}}))$ and 
		$\img(\obs_{\alpha(\exp_S(\ext(\mathcal{X})))}) = \img(h \circ S(\obs^{\dag}_{\mathcal{X}}))$, where $\langle \Omega, h, k \rangle$ is the final $\lambda^S$-bialgebra.
	\end{lemma}

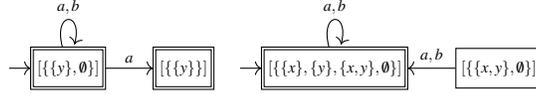
\begin{figure}[t]
\tiny
\centering
	\begin{tikzpicture}[node distance=3em]
		\node[state, initial, shape=rectangle, initial text=, accepting] (7) {$\lbrack \lbrace \lbrace y \rbrace, \emptyset \rbrace \rbrack$};
		
	\node[state, shape=rectangle, right = of 7, accepting] (6) {$\lbrack \lbrace \lbrace y \rbrace \rbrace \rbrack$};

	\node[state, shape=rectangle, initial, right = of 6, accepting, initial text=] (8) {$\lbrack \lbrace \lbrace x \rbrace, \lbrace y \rbrace, \lbrace x, y \rbrace, \emptyset \rbrace \rbrack$};
	
		\node[state, shape=rectangle, right = of 8] (4) {$\lbrack \lbrace \lbrace x, y \rbrace, \emptyset \rbrace \rbrack$};
		
		    \path[->]
	(7) edge[loop above] node{$a,b$} (7)
	(7) edge[above] node{$a$} (6)
	(8) edge[loop above] node{$a,b$} (8)
	(4) edge[above] node{$a,b$} (8)
;
	\end{tikzpicture}
	\caption{The minimal xor-CABA automaton for $\mathcal{L} = (a+b)^*a$.}
			\label{minimalxorcaba}
\end{figure}

\subsection{Applications to canonical automata}

\label{minimalityimplications}

In this section we instantiate \Cref{minimalitytheorem} to characterise a variety of canonical acceptors from the literature as size-minimal representatives among subclasses of $\alpha$-closed succinct automata, i.e. those automata whose images of the two semantics induced by $\alpha$ coincide.
We begin with the canonical RFSA and the minimal xor automaton, for which $\alpha$ is the identitity and $\alpha$-closedness therefore is trivial. 

In \cite{DenisLT02} the canonical RFSA for $\mathcal{L}$ has been characterised as size-minimal among those NFAs accepting $\mathcal{L}$ for which states accept a residual of $\mathcal{L}$. More recently, it was shown that the class in fact can be extended to those NFAs accepting $\mathcal{L}$ for which states accept a \emph{union} of residuals of $\mathcal{L}$ \cite{MyersAMU15}. The next result recovers the latter as a consequence of the second point in \Cref{minimalitytheorem}. We write $\overline{Y}$ for the algebraic closure\footnote{If $Y = \textnormal{im}(f)$ for some morphism $f$ with codomain $\langle X,h \rangle$, the closure is given by the induced $T$-algebra structure on $\textnormal{im}(h \circ Tf)$.} of a subset $Y \subseteq X$ of some $T$-algebra $X$.

\begin{corollary}
\label[corollary]{canonicalrfsaminimal}
	The canonical RFSA for $\mathcal{L}$ is size-minimal among non-deterministic automata $\mathcal{Y}$ accepting $\mathcal{L}$ with 
	$
	 \overline{\textnormal{im}(\obs^{\dag}_{\mathcal{Y}})}^{\textnormal{CSL}} \subseteq \overline{\textnormal{Der}(\mathcal{L})}^{\textnormal{CSL}} 
	$.
\end{corollary}

The second condition in \Cref{minimalitytheorem} is always satisfied for a \emph{reachable} succinct automaton $\mathcal{Y}$. Since for $\mathbb{Z}_2$-weighted automata it is possible to find an equivalent reachable $\mathbb{Z}_2$-weighted automaton with less or equally many states (which for NFA is not necessarily the case), the minimal xor automaton is minimal among \emph{all} $\mathbb{Z}_2$-weighted automata, as was already known from for instance \cite{VuilleminG210}.

\begin{corollary}
\label[corollary]{minimalxor}
		The minimal xor automaton for $\mathcal{L}$ is size-minimal among $\mathbb{Z}_2$-weighted automata accepting $\mathcal{L}$.
\end{corollary}

For the \'atomaton, the distromaton, and the minimal xor-CABA automaton the distributive law homomorphism $\alpha$ in play is non-trivial; $\alpha$-closedness translates to the below equalities between closures. In all three cases it is possible to waive the inclusion induced by the second point in \Cref{minimalitytheorem}.

\begin{corollary}
\label[corollary]{minimalityatomaton}
		The \'atomaton for $\mathcal{L}$ is size-minimal among non-deterministic automata $\mathcal{Y}$ accepting $\mathcal{L}$ with $\overline{\textnormal{im}(\obs^{\dag}_{\mathcal{Y}})}^{\textnormal{CSL}} = \overline{\textnormal{im}(\obs^{\dag}_{\mathcal{Y}})}^{\textnormal{CABA}}$.
\end{corollary}

The above result can be shown to be similar to \cite[Theorem 4.9]{MyersAMU15}, which characterises the \'atomaton as size-minimal among non-deterministic automata whose accepted languages are \emph{closed under complement}. The result below is very similar to a characterisation of the distromaton as size-minimal among non-deterministic automata whose accepted languages are \emph{closed under intersection} \cite[Theorem 4.13]{MyersAMU15}.

\begin{corollary}
\label[corollary]{distromatonminimal}
		The distromaton for $\mathcal{L}$ is size-minimal among non-deterministic automata $\mathcal{Y}$ accepting $\mathcal{L}$ with $ \overline{\textnormal{im}(\obs^{\dag}_{\mathcal{Y}})}^{\textnormal{CSL}} = \overline{\textnormal{im}(\obs^{\dag}_{\mathcal{Y}})}^{\textnormal{CDL}}$.
\end{corollary}

The size-minimality result for the newly discovered minimal xor-CABA automaton is analogous to the ones for the \'atomaton and the distromaton.

\begin{corollary}
\label[corollary]{corminimalxorcaba}
		The minimal xor-CABA automaton for $\mathcal{L}$ is size-minimal among $\mathbb{Z}_2$-weighted automata $\mathcal{Y}$ accepting $\mathcal{L}$ with $ \overline{\textnormal{im}(\obs^{\dag}_{\mathcal{Y}})}^{\mathbb{Z}_2\textnormal{-Vect}} = \overline{\textnormal{im}(\obs^{\dag}_{\mathcal{Y}})}^{\textnormal{CABA}}$.
\end{corollary}

We conclude with a size-comparison between acceptors that is parametric in the closure of derivatives.

\begin{corollary}
\label[corollary]{sizecomparison}
	\begin{itemize}
		\item If $\overline{\textnormal{Der}(\mathcal{L})}^{\mathbb{Z}_2\textnormal{-Vect}} = \overline{\textnormal{Der}(\mathcal{L})}^{\textnormal{CABA}}$, then the minimal xor automaton and the minimal xor-CABA automaton for $\mathcal{L}$ are of the same size.
		\item If $\overline{\textnormal{Der}(\mathcal{L})}^{\textnormal{CSL}} = \overline{\textnormal{Der}(\mathcal{L})}^{\textnormal{CDL}}$, then the canonical RFSA and the distromaton for $\mathcal{L}$ are of the same size.
		\item If $\overline{\textnormal{Der}(\mathcal{L})}^{\textnormal{CSL}} = \overline{\textnormal{Der}(\mathcal{L})}^{\textnormal{CABA}}$, then the canonical RFSA and the \'atomaton for $\mathcal{L}$ are of the same size. 
	\end{itemize}
\end{corollary}

\section{Related work}

\label{relatedwork}

One of the main motivations for the present paper is provided by active learning algorithms for the derivation of succinct state-based models \cite{angluin1987learning}. A major challenge in learning non-deterministic models is the lack of a canonical target acceptor for a given language \cite{DenisLT02}. The problem has been independently approached for different variants of non-determinism, often with the idea of finding a subclass admitting a unique representative \cite{esposito2002learning,berndt2017learning} such as e.g. the canonical RFSA, the minimal xor automaton, or the \'atomaton. 

 A more general and unifying perspective on learning automata that may not have a canonical target was given by Van Heerdt \cite{van2017learning, van2016master, van2020phd}. One of the central notions in this work is the concept of a scoop, originally introduced by Arbib and Manes \cite{arbib1975fuzzy} and here referred to as a generator. The main contribution in \cite{van2017learning} is a general procedure to find irreducible sets of generators, which thus restricts the work to the category of sets. In the present paper we generally work over arbitrary categories, although we assume the existence of a minimal set-based generator in \Cref{minimalitytheorem}. Furthermore, the work of Van Heerdt has no size-minimality results.
 

Closely related to the present paper is the work of Myers et al.\ \cite{MyersAMU15}, who present a coalgebraic construction for canonical non-deterministic automata. They cover the canonical RFSA, the minimal xor automaton, the \'atomaton, and the distromaton. The underlying idea in \cite{MyersAMU15} for finding succinct representations is similar to ours: first they build the minimal DFA for a regular language in a locally finite variety, then they apply an equivalence between the category of finite algebras and a suitable category of finite structured sets and relations.
On the one hand, the category of finite algebras in a locally finite variety can be translated into our setting by considering a category of algebras over a monad preserving finite sets. In fact, modulo this translation, many of the categories considered here already appear in \cite{MyersAMU15}, e.g.\ vector spaces, Boolean algebras, complete join-semi lattices, and distributive lattices. On the other hand, their construction seems to be restricted to the category of sets and non-deterministic automata, while we work over arbitrary monads on arbitrary categories. Their work does not provide a general algorithm to construct a succinct automaton, i.e., the specifics vary with the equivalences considered, while we give a general definition and a soundness argument in \Cref{generatorbialgebrahom}. While Myers et al.\ give minimality results for a wide range of acceptors, each proof follows case-specific arguments. In \Cref{minimalitytheorem} we provide a unifying minimality result for succinct automata that generalises parts of \cite[Theorem 4.8]{MyersAMU15} and subsumes most of their results \cite[Theorem 4.9, Theorem 4.10, Corollary 4.11, Theorem 4.13]{MyersAMU15}.
\section{Discussion and future work}

\label{discussion}

We have presented a general categorical framework based on bialgebras and distributive law homomorphisms for the derivation of canonical automata. The framework instantiates to a wide range of well-known examples from the literature and allowed us to discover a previously unknown canonical acceptor for regular languages. Finally, we presented a theorem that subsumes previously independently proven minimality results for canonical acceptors, implied new characterisations, and allowed us to make size-comparisons between canonical automata. 

In the future, we would like to cover other examples, such as the canonical probabilistic RFSA \cite{esposito2002learning} and the canonical alternating RFSA \cite{berndt2017learning, angluin2015learning}.
Probabilistic automata of the type in \cite{esposito2002learning} are typically modelled as $TF$-coalgebras instead of $FT$-coalgebras \cite{jacobs2012trace}, and thus will need a shift in perspective.
For alternating RFSAs we expect a canonical form can be constructed in the spirit of this paper, from generators for algebras over the neighbourhood monad, by interpreting the join-dense atoms of a CABA as a full meet of ground elements. 

Generally, it would be valuable to have a more systematic treatment of the range of available monads and distributive law homomorphisms \cite{zwart2019no}, making use of the fact that distributive law homomorphisms compose.

Further generalisation in another direction could be achieved by distributive laws between monads and endofunctors on different categories. For instance, we expect that operations on automata as the product can be captured by homomorphisms between distributive laws of such more general type. 

Finally, we would like to lift existing double-reversal characterisations of the minimal DFA \cite{brzozowski1962canonical}, the \'atomaton \cite{BrzozowskiT14}, the distromaton \cite{MyersAMU15}, and the minimal xor automaton \cite{VuilleminG210} to general canonical automata. The work in \cite{bonchi2012brzozowski, bonchi2014algebra} gives a coalgebraic generalisation of Brzozowski's algorithm based on dualities between categories, but does not cover the cases we are interested in. The framework in \cite{adamek2012coalgebraic} recovers the \'atomaton as the result of a minimisation procedure, but does not consider other canonical acceptors. 
\bibliographystyle{eptcs}
\bibliography{generalisingrfsa}

\ifdefined\doappendix
\label{appendix}

\section{Definitions}

\begin{definition}
\label[definition]{monaddefs}
\begin{itemize}
	\item The \emph{powerset monad} $\langle \mathcal{P}, \eta^{\mathcal{P}}, \mu^{\mathcal{P}} \rangle$ on the category of sets is defined by 
	$\mathcal{P}X = 2^X$, $\mathcal{P}f(\varphi)(y) = \bigvee_{x \in f^{-1}(y)} \varphi(x)$, $\eta_X^{\mathcal{P}}(x)(y) = \lbrack x = y \rbrack$, and $\mu^{\mathcal{P}}_X(\Phi)(x) = \bigvee_{\varphi \in 2^X} \Phi(\varphi) \wedge \varphi(x)$.
	\item The \emph{neighbourhood monad} $\langle \mathcal{H}, \eta^{\mathcal{H}}, \mu^{\mathcal{H}} \rangle$ on the category of sets is defined by
	$\mathcal{H}X = 2^{2^X}$, $\mathcal{H}f(\Phi)(\varphi) = \Phi(\varphi \circ f)$, $\eta_X^{\mathcal{H}}(x)(\varphi) = \varphi(x)$, and $\mu^{\mathcal{H}}_X(\Psi)(\varphi) = \Psi(\eta^{\mathcal{H}}_{2^X}(\varphi))$.
	\item The \emph{monotone neighbourhood monad} $\langle \mathcal{A}, \eta^{\mathcal{A}}, \mu^{\mathcal{A}} \rangle$ on the category of sets is defined by $\mathcal{A}X =  \langle 2, \leq \rangle^{\langle 2^X, \subseteq \rangle}$ and otherwise coincides with the neighbourhood monad.
	\item The \emph{free vector space monad} over the unique two element field $\langle \mathcal{R},  \eta^{\mathcal{R}}, \mu^{\mathcal{R}} \rangle$ on the category of sets is defined by $\mathcal{R}X = 2^X$, $\mathcal{R}f(\varphi)(y) = \bigoplus_{x \in f^{-1}(y)} \varphi(x)$, where $a \oplus b := a + b \textnormal{ mod } 2$, $\eta_X^{\mathcal{R}}(x)(y) = \lbrack x = y \rbrack$, and $\mu^{\mathcal{R}}_X(\Phi)(x) = \bigoplus_{\varphi \in 2^X} \Phi(\varphi) \cdot \varphi(x)$.
\item The \emph{nominal powerset monad} $\langle \mathcal{P}_{\textnormal{n}}, \eta^{\mathcal{P}_{\textnormal{n}}}, \mu^{\mathcal{P}_{\textnormal{n}}} \rangle$ on the category of (finitely-supported) nominal sets is defined by $\mathcal{P}_{\textnormal{n}}X = \lbrace A \subseteq X \mid A \textnormal{ finitely supported} \rbrace$, $\pi.A := \lbrace \pi.a \mid a \in A \rbrace$, and otherwise coincides with $\mathcal{P}$ \cite{pitts2013nominal}. \end{itemize}
\end{definition}

\section{Proofs}

\begin{lemma}\tagcite{jacobs2005bialgebraic}
\label[lemma]{induceddistrlaw}
	Every algebra $h: TB \rightarrow B$ for a set monad $T$ induces a distributive law $\lambda^h$ between $T$ and $F$ with $FX = B \times X^A $ by $\lambda^h_X := (h \times \textnormal{st}) \circ \langle T \pi_1, T \pi_2 \rangle$.
\end{lemma}
\begin{proof}
The statement is well-known \cite{jacobs2005bialgebraic}, but a complete proof hard to find. The naturality of $\lambda^h$ essentially follows from the naturality of the strength function.
The equation involving the monad unit is a consequence of
\begin{align*}
	& \pi_1 \circ (h \times \textnormal{st}) \circ \langle T\pi_1, T\pi_2 \rangle \circ \eta_{B \times X^A} & \\
=\ & h \circ T\pi_1 \circ \eta_{B \times X^A} & \textnormal{(Def. } \pi_1) \\
=\ & h \circ \eta_B \circ \pi_1 & \textnormal{(Nat. } \eta) \\
=\ &  \pi_1 \circ (B \times \eta_X^A) & (h \circ \eta_B = \textnormal{id}_B, \textnormal{Def. } \pi_1)
\end{align*}
and 
\begin{align*}
	& \pi_2 \circ (h \times \textnormal{st}) \circ \langle T\pi_1, T\pi_2 \rangle \circ \eta_{B \times X^A} & \\
	=\ & \textnormal{st} \circ T\pi_2 \circ \eta_{B \times X^A} & \textnormal{(Def. } \pi_2) \\
	=\ & \textnormal{st} \circ \eta_{X^A} \circ \pi_2 & \textnormal{(Nat. } \eta) \\
	=\ & \eta_X^A \circ \pi_2 & (\textnormal{st} \circ \eta_{X^A} = \eta_X^A)  \\
	=\ & \pi_2 \circ (B \times \eta_X^A) & \textnormal{(Def. } \pi_2).
\end{align*}
	Similarly, the equation involving the monad multiplication is a consequence of
		\begin{align*}
		&\pi_1 \circ (B \times \mu_X^A) \circ (h \times \textnormal{st}) \circ \langle T \pi_1, T \pi_2 \rangle 
		 \circ T(h \times \textnormal{st}) \circ T\langle T\pi_1, T\pi_2 \rangle  \\
		 =\ & h \circ Th \circ T^2\pi_1 & \textnormal{(Def. } \pi_1) \\
		  =\ & h \circ \mu_B \circ T^2 \pi_1 & (h \circ Th = h \circ \mu_B) \\
		  =\ & h \circ T \pi_1 \circ \mu_{B \times X^A} & \textnormal{(Nat. } \mu) \\
		  =\ & \pi_1 \circ (h \times \textnormal{st}) \circ \langle T\pi_1, T\pi_2 \rangle \circ \mu_{B \times X^A} & \textnormal{(Def. } \pi_1)
	\end{align*}
	and
	\begin{align*}
		&\pi_2 \circ (B \times \mu_X^A) \circ (h \times \textnormal{st}) \circ \langle T \pi_1, T \pi_2 \rangle 
		 \circ T(h \times \textnormal{st}) \circ T\langle T\pi_1, T\pi_2 \rangle  \\
		=\ & \mu_X^A \circ \textnormal{st} \circ T(\textnormal{st}) \circ T^2\pi_2 & \textnormal{(Def. } \pi_2) \\
		=\ & \textnormal{st} \circ \mu_{X^A} \circ T^2\pi_2 & (\mu^A_X \circ \textnormal{st} \circ T(\textnormal{st}) = \textnormal{st} \circ \mu_{X^A}) \\
		=\ & \textnormal{st} \circ T\pi_2 \circ \mu_{B \times X^A} & \textnormal{(Nat. } \mu) \\
		=\ & \pi_2 \circ (h \times \textnormal{st}) \circ \langle T\pi_1, T\pi_2 \rangle \circ \mu_{B \times X^A} & \textnormal{(Def. } \pi_2).
	\end{align*}
\end{proof}

\begin{lemma}
\label[lemma]{powersetalgebra}
	The morphism $h^{\mathcal{P}}: \mathcal{P}2 \rightarrow 2$ satisfying $\varphi \mapsto \varphi(1)$ defines a $\mathcal{P}$-algebra.
\end{lemma}
\begin{proof}
On the one hand we find
	\begin{align*}
		h^{\mathcal{P}} \circ \mu^{\mathcal{P}}_2(\Phi) &= \mu^{\mathcal{P}}_2(\Phi)(1) & \textnormal{(Def. } h^{\mathcal{P}}) \\
		&= \bigvee_{\varphi \in 2^2} \Phi(\varphi) \wedge \varphi(1) & \textnormal{(Def. } \mu^{\mathcal{P}}_2) \\
		&= \bigvee_{\varphi \in (h^{\mathcal{P}})^{-1}(1)} \Phi(\varphi) & \textnormal{(Def. } h^{\mathcal{P}}) \\
		&= \mathcal{P}(h^{\mathcal{P}})(\Phi)(1) & \textnormal{(Def. } \mathcal{P}(h^{\mathcal{P}})) \\
		&= h^{\mathcal{P}} \circ \mathcal{P}(h^{\mathcal{P}})(\Phi) & \textnormal{(Def. } h^{\mathcal{P}}),
	\end{align*}
	and on the other hand we can deduce
	\begin{align*}
		h^{\mathcal{P}} \circ \eta^{\mathcal{P}}_2(x) &= \eta^{\mathcal{P}}_2(x)(1) & \textnormal{(Def. } h^{\mathcal{P}}) \\
		&= \lbrack x = 1 \rbrack & \textnormal{(Def. } \eta^{\mathcal{P}}_2 ) \\
		&= x & \textnormal{(Def. } \lbrack \cdot \rbrack).
	\end{align*}
\end{proof}

\begin{lemma}
\label[lemma]{neighbourhoodalgebra}
	The morphism $h^{\mathcal{H}}: \mathcal{H}2 \rightarrow 2$ assigning $\Phi \mapsto \Phi(\textnormal{id}_2)$ defines a $\mathcal{H}$-algebra.
\end{lemma}
\begin{proof}
	Since $\eta^{\mathcal{H}}_{2^2}(\textnormal{id}_2)(\Phi) = \Phi(\textnormal{id}_2) = h^{\mathcal{H}}(\Phi)$ we find
	\begin{align*}
		h^{\mathcal{H}} \circ \mu^{\mathcal{H}}_2(\Psi) &= \mu_2^{\mathcal{H}}(\Psi)(\textnormal{id}_2) & \textnormal{(Def. } h^{\mathcal{H}}) \\
		&= \Psi(\eta_{2^2}^{\mathcal{H}}(\textnormal{id}_2)) & (\textnormal{Def. } \mu^{\mathcal{H}}_2) \\
		&= \Psi(\textnormal{id}_2 \circ h^{\mathcal{H}}) & (\eta^{\mathcal{H}}_{2^2}(\textnormal{id}_2) = h^{\mathcal{H}})  \\
		&= \mathcal{H}(h^{\mathcal{H}})(\Psi)(\textnormal{id}_2) & \textnormal{(Def. } \mathcal{H}(h^{\mathcal{H}})) \\
		&= h^{\mathcal{H}} \circ \mathcal{H}(h^{\mathcal{H}})(\Psi) & (\textnormal{Def. } h^{\mathcal{H}}).
	\end{align*}
	We further can deduce
	\begin{align*}
		h^{\mathcal{H}} \circ \eta_X^{\mathcal{H}}(x) &= \eta_X^{\mathcal{H}}(x)(\textnormal{id}_2) & \textnormal{(Def. } h^{\mathcal{H}}) \\
		&= \textnormal{id}_2(x) & \textnormal{(Def. } \eta^{\mathcal{H}}_X) \\
		&= x & \textnormal{(Def. } \textnormal{id}_2).
	\end{align*}
\end{proof}

\begin{lemma}
\label[lemma]{monotoneneighbourhoodoutputalgebra}
	The morphism $h^{\mathcal{A}}: \mathcal{A}2 \rightarrow 2$ assigning $\Phi \mapsto \Phi(\textnormal{id}_2)$ defines a $\mathcal{A}$-algebra.
\end{lemma}
 \begin{proof}
 	Analogous to the proof of \Cref{neighbourhoodalgebra}.
 \end{proof}
 
 \begin{lemma}
\label[lemma]{xoroutputalgebra}
	The morphism $h^{\mathcal{R}}: \mathcal{R}2 \rightarrow 2$ satisfying $\varphi \mapsto \varphi(1)$ defines a $\mathcal{R}$-algebra.
\end{lemma}
\begin{proof}
	Analogous to the proof of \Cref{powersetalgebra}.
\end{proof}

\begin{oneshot}{proposition}{forgenerator-isharp-is-bialgebra-hom}
	Let $\langle X, h, k\rangle$ be a $\lambda$-bialgebra and let $\langle Y, i, d \rangle$ be a generator for the $T$-algebra $\langle X,h \rangle$.
	Then $h \circ Ti \colon \textnormal{exp}_T(\langle Y, Fd \circ k \circ i\rangle ) \rightarrow \langle X, h, k \rangle$ is a $\lambda$-bialgebra homomorphism.
\end{oneshot}
\begin{proof}
By definition we have $\textnormal{exp}_T(\langle Y, Fd \circ k \circ i \rangle) = \langle TY, \mu_Y, F \mu_Y \circ \lambda_{TY} \circ T(Fd \circ k \circ i) \rangle$. It is well-known that $h \circ Ti$ is a  homomorphism between the underlying $T$-algebra structures. It thus remains to show that it is a $F$-coalgebra homomorphism. The latter follows from the commutativity of the diagram below:
\begin{equation*}
	\begin{tikzcd}		TY \arrow{r}{Ti} \arrow{d}[left]{Ti} & TX \arrow{dd}{Tk} \arrow{rr}{h} & & X \arrow{ddddd}{k} \\
		TX \arrow{d}[left]{Tk} \\
		TFX \arrow{d}[left]{TFd} \arrow{r}{\textnormal{id}_{TFX}} & TFX \arrow{ddr}{\lambda_X} \\
		TFTY \arrow{d}[left]{\lambda_{TY}} \arrow{r}{TFTi} & TFTX \arrow{u}{TFh} \\
		FT^2Y \arrow{r}{FT^2i} \arrow{d}[left]{F\mu_Y} & FT^2X \arrow{r}{FTh} \arrow{d}{F \mu_X} & FTX \arrow{d}{Fh} \\
		FTY \arrow{r}{FTi} & FTX \arrow{r}{Fh} & FX \arrow{r}{\textnormal{id}_{FX}} & FX
	\end{tikzcd}.
\end{equation*}
\end{proof}

\begin{oneshot}{proposition}{forbasis-isharp-is-bialgebra-iso}
	Let $\langle X, h, k\rangle$ be a $\lambda$-bialgebra and let $\langle Y, i, d \rangle$ be a basis for the $T$-algebra $\langle X,h \rangle$.
	Then $h \circ Ti \colon \textnormal{exp}_T(\langle Y, Fd \circ k \circ i\rangle ) \rightarrow \langle X, h, k \rangle$ is a $\lambda$-bialgebra isomorphism.
\end{oneshot}
\begin{proof}
By definition we have $\textnormal{exp}_T(\langle Y, Fd \circ k \circ i \rangle) = \langle TY, \mu_Y, F \mu_Y \circ \lambda_{TY} \circ T(Fd \circ k \circ i) \rangle$. From \Cref{forgenerator-isharp-is-bialgebra-hom} we know that $h\circ Ti$ is a $\lambda$-bialgebra homomorphism. By the definition of a basis, $d$ is a two-sided inverse to $h\circ Ti$ as ordinary morphism. It thus remains to show that $d$ is a $\lambda$-bialgebra homomorphism. The diagram below on the left shows that it is a $T$-algebra homomorphism, and the diagram on the right below shows that it commutes with $F$-coalgebra structures: 
\begin{equation*}
\begin{tikzcd}[column sep=small]
				TX \arrow{rr}{Td} \arrow{d}[right]{\textnormal{id}_{TX}} & & T^2Y \arrow{dl}{T^2i} \arrow{dd}[left]{\mu_Y} \\
				TX \arrow{dd}[right]{h} & T^2X \arrow{l}{Th} \arrow{d}{\mu_X} \\
				& TX \arrow{dl}{h} & TY \arrow{l}{Ti} \arrow{d}[left]{\textnormal{id}_{TY}} \\
				X \arrow{rr}{d} & & TY
			\end{tikzcd}
			\quad
		\begin{tikzcd}
			X \arrow{rrrrr}{k} \arrow{dd}[right]{d} \arrow{rd}{\textnormal{id}_X} & & & & & FX \arrow{dd}[left]{Fd} \\
			& X \arrow{rrrru}{k} & & & FTX \arrow{d}{FTd} \arrow{ur}{Fh} & \\
			TY \arrow{r}[below]{Ti} & TX \arrow{u}{h} \arrow{r}[below]{Tk} & TFX \arrow{r}[below]{TFd} \arrow{rru}{\lambda_X} & TFTY \arrow{r}[below]{\lambda_{TY}} & FT^2Y \arrow{r}[below]{F\mu_Y} & FTY
		\end{tikzcd}.
\end{equation*}
\end{proof}

 \begin{lemma}
 \label[lemma]{equivariantstrength}
 	The strength function $\textnormal{st}: \mathcal{P}_{\textnormal{n}}(X^A) \rightarrow (\mathcal{P}_{\textnormal{n}}X)^A$ satisfying $ \textnormal{st}(\Phi)(a) = \mathcal{P}_{\textnormal{n}}(\textnormal{ev}_a)(\Phi)$ is equivariant.
 \end{lemma}
 \begin{proof}
 We first observe that for any $a \in A, x \in X$ and $\pi \in \textnormal{Perm}(A)$ the mapping
 \begin{equation}
 	\label{eq:strengthequivariant}
 \lbrace \varphi \in X^A \mid \varphi(a) = x \rbrace \rightarrow \lbrace \varphi \in X^A \mid \varphi(\pi^{-1}.a) = \pi^{-1}.x \rbrace \qquad \pi \mapsto \pi^{-1}.\varphi
 \end{equation}
 defines a bijection with inverse assignment $\varphi \mapsto \pi.\varphi$. Note that the set $2$ is equipped with the trivial action.
 	The statement thus follows from
 	\begin{align*}
 		(\pi.\textnormal{st}(\Phi))(a)(x) 
 		&= \pi.(\textnormal{st}(\Phi)(\pi^{-1}.a))(x) & \textnormal{(Def. } \pi.\textnormal{st}(\Phi)) \\
 		&= \textnormal{st}(\Phi)(\pi^{-1}.a)(\pi^{-1}.x) & \textnormal{(Def. } \pi.(\textnormal{st}(\Phi)(\pi^{-1}.a)) ) \\
 		&= \mathcal{P}_{\textnormal{n}}(\textnormal{ev}_{\pi^{-1}.a})(\Phi)(\pi^{-1}.x) & \textnormal{(Def. st)} \\
 		&= \bigvee_{\varphi \in \textnormal{ev}_{\pi^{-1}.a}^{-1}(\pi^{-1}.x)} \Phi(\varphi) & \textnormal{(Def. } \mathcal{P}_{\textnormal{n}}(\textnormal{ev}_{\pi^{-1}.a})) \\
 		&= \bigvee_{\varphi \in \textnormal{ev}_{a}^{-1}(x)} \Phi(\pi^{-1}.\varphi) & \textnormal{\eqref{eq:strengthequivariant}} \\
 		&= \bigvee_{\varphi \in \textnormal{ev}_a^{-1}(x)} (\pi.\Phi)(\varphi) & \textnormal{(Def. } \pi.\Phi) \\
 		&= \mathcal{P}_{\textnormal{n}}(\textnormal{ev}_a)(\pi.\Phi)(x) & \textnormal{(Def. } \mathcal{P}_{\textnormal{n}}(\textnormal{ev}_a)) \\
 		&= \textnormal{st}(\pi.\Phi)(a)(x) & \textnormal{(Def. st)}.
 	\end{align*}
 \end{proof}

\begin{lemma}
\label[lemma]{freepowersetbialgebrastructure}
		Let $\langle \mathcal{P}X, \mu^{\mathcal{P}}_X, \langle \overline{\varepsilon}, \overline{\delta} \rangle \rangle := \textnormal{free}^{\lambda^{\mathcal{P}}}(\langle X, \langle \varepsilon, \delta \rangle \rangle)$. Then $\overline{\varepsilon}(\varphi) = \bigvee_{y \in \varepsilon^{-1}(1)} \varphi(y)$ and $\overline{\delta}_a(\varphi)(x) = \bigvee_{y \in \delta_a^{-1}(x)} \varphi(y)$.
\end{lemma}
\begin{proof}
	The first equality is a consequence of 
	\begin{align*}
		\overline{\varepsilon}(\varphi) 
		&= \pi_1 \circ (h^{\mathcal{P}} \times \textnormal{st}) \circ \langle \mathcal{P}\pi_1, \mathcal{P}\pi_2 \rangle \circ \mathcal{P}(\langle \varepsilon, \delta \rangle)(\varphi) & \textnormal{(Def. } \overline{\varepsilon}) \\
		&= h^{\mathcal{P}} \circ \mathcal{P}(\varepsilon)(\varphi) & \textnormal{(Def. } \pi_1) \\
		&= \mathcal{P}(\varepsilon)(\varphi)(1) & \textnormal{(Def. } h^{\mathcal{P}}) \\
		&= \bigvee_{y \in \varepsilon^{-1}(1)} \varphi(y) & \textnormal{(Def. } \mathcal{P}(\varepsilon)).
	\end{align*}
	For the second equality we observe
	\begin{align*}
		\overline{\delta}_a(\varphi)(x) &= 
		\overline{\delta}(\varphi)(a)(x) & \textnormal{(Def. } \overline{\delta}_a) \\
		&= \pi_2 \circ (h^{\mathcal{P}} \times \textnormal{st} ) \circ \langle \mathcal{P}\pi_1, \mathcal{P}\pi_2 \rangle \circ \mathcal{P}(\langle \varepsilon, \delta \rangle)(\varphi)(a)(x) & \textnormal{(Def. } \overline{\delta}) \\
		&= \textnormal{st} \circ \mathcal{P}(\delta)(\varphi)(a)(x) & \textnormal{(Def. } \pi_2) \\
		&= \mathcal{P}(\textnormal{ev}_a)(\mathcal{P}(\delta)(\varphi))(x) &\textnormal{(Def. st)} \\
		&= \mathcal{P}(\delta_a)(\varphi)(x) & (\textnormal{Def. } \delta_a) \\
		&=  \bigvee_{y \in \delta_a^{-1}(x)} \varphi(y) & \textnormal{(Def. } \mathcal{P}(\delta_a)).
	\end{align*}
\end{proof}

\begin{lemma}
\label[lemma]{freeneighbourhoodbialgebrastructure}
	Let $\langle \mathcal{H}X, \mu^{\mathcal{H}}_X, \langle  \overline{\varepsilon}, \overline{\delta} \rangle \rangle := \textnormal{free}^{\lambda^{\mathcal{H}}}(\langle X, \langle \varepsilon, \delta \rangle \rangle)$. Then $\overline{\varepsilon}(\Phi) = \Phi(\varepsilon)$ and $\overline{\delta}_a(\Phi)(\varphi) = \Phi(\varphi \circ \delta_a)$.
\end{lemma}
\begin{proof}
The proof is analogous to the one of \Cref{freepowersetbialgebrastructure}.
	The first equality is a consequence of 
	\begin{align*}
		\overline{\varepsilon}(\Phi) 
		&= \mathcal{H}(\varepsilon)(\Phi)(\textnormal{id}_2)  & \textnormal{(Cf. proof of \Cref{freepowersetbialgebrastructure})} \\
		&= \Phi(	\textnormal{id}_2 \circ \varepsilon) & \textnormal{(Def. } \mathcal{H}(\varepsilon)) \\
		&= \Phi(\varepsilon) & (\textnormal{id}_2 \circ \varepsilon = \varepsilon).
	\end{align*}
		For the second equality we observe
	\begin{align*}
		\overline{\delta}_a(\Phi)(\varphi) 
		&= \mathcal{H}(\delta_a)(\Phi)(\varphi) & \textnormal{(Cf. proof of \Cref{freepowersetbialgebrastructure})} \\
		&= \Phi(\varphi \circ \delta_a) & \textnormal{(Def. } \mathcal{H}(\delta_a)).
	\end{align*}
\end{proof}

\begin{lemma}
\label[lemma]{freealternatingbialgebrastructure}
	Let $\langle \mathcal{A}X, \mu^{\mathcal{A}}_X, \langle \overline{\varepsilon}, \overline{\delta}  \rangle \rangle := \textnormal{free}^{\lambda^{\mathcal{A}}}(\langle X, \langle \varepsilon, \delta  \rangle \rangle)$. Then $\overline{\varepsilon}(\Phi) = \Phi(\varepsilon)$ and $\overline{\delta}_a(\Phi)(\varphi) = \Phi(\varphi \circ \delta_a)$.
\end{lemma}
\begin{proof}
	Analogous to the proof of \Cref{freeneighbourhoodbialgebrastructure}.
\end{proof}

\begin{lemma}
\label[lemma]{freexorbialgebrastructure}
		Let $\langle \mathcal{R}X, \mu^{\mathcal{R}}_X, \langle \overline{\varepsilon}, \overline{\delta} \rangle \rangle := \textnormal{free}^{\lambda^{\mathcal{R}}}(\langle X, \langle \varepsilon, \delta \rangle \rangle)$. Then $\overline{\varepsilon}(\varphi) = \bigoplus_{y \in \varepsilon^{-1}(1)} \varphi(y)$ and $\overline{\delta}_a(\varphi)(x) = \bigoplus_{y \in \delta_a^{-1}(x)} \varphi(y)$.
\end{lemma}
\begin{proof}
	Analogous to the proof of \Cref{freepowersetbialgebrastructure}.
\end{proof}

\begin{oneshot}{lemma}{inducedbialgebra}
\tagcite{klin2015presenting, bonsangue2013presenting}
Let $\alpha: \lambda^S \rightarrow \lambda^T$ be a distributive law homomorphism. Then $\alpha \langle X, h, k\rangle := \langle X, h \circ \alpha_X, k \rangle$ and  $\alpha(f) := f$ defines a functor $\alpha: \textnormal{Bialg}(\lambda^S) \rightarrow \textnormal{Bialg}(\lambda^T)$.
\end{oneshot}
\begin{proof}
The statement is well-known \cite{klin2015presenting, bonsangue2013presenting}, but a complete proof difficult to find. We first show that the definition is well-defined on objects.
The commutativity of the two diagrams below shows that $\langle X, h \circ \alpha_X \rangle$ is a $T$-algebra:
	\begin{equation*}
		\begin{tikzcd}
			T^2X \arrow{rr}{\mu^T_X} \arrow{d}[left]{T\alpha_X} & & TX \arrow{d}{\alpha_X} \\
			TSX \arrow{d}[left]{Th} \arrow{r}{\alpha_{SX}} & S^2X \arrow{r}{\mu^S_X} \arrow{d}{Sh} & SX  \arrow{d}{h}\\
			TX \arrow{r}[below]{\alpha_X} & SX \arrow{r}[below]{h} & X
		\end{tikzcd}
\qquad
	\begin{tikzcd}
		X \arrow{rr}{1} \arrow{dr}{\eta_X^S} \arrow{dd}[left]{\eta^T_X} & & X \\
		& SX \arrow{ur}[right]{h} & \\
		TX \arrow{ur}[right]{\alpha_X} & & 
	\end{tikzcd}.
	\end{equation*}
	To establish that $\langle X, h\circ \alpha_X, k \rangle$ is a $\lambda^T$-bialgebra it thus remains to observe the commutativity of the diagram on the left below:
	\begin{equation*}
		\begin{tikzcd}
		TX \arrow{rr}{Tk} \arrow{dd}[left]{\alpha_X} & & TFX \arrow{d}{\lambda^T_X} \arrow{ddl}[left]{\alpha_{FX}} \\
		& & FTX \arrow{d}{F\alpha_X} \\
		SX \arrow{d}[left]{h} \arrow{r}{Sk} & SFX \arrow{r}{\lambda^S_X} & FSX \arrow{d}{Fh} \\
		X \arrow{rr}[below]{k} & & FX	
		\end{tikzcd}
		\qquad
				\begin{tikzcd}
		TX \arrow{d}[left]{\alpha_X} \arrow{r}{Tf} &TY  \arrow{d}{\alpha_Y}	\\
		SX \arrow{d}[left]{h_X} \arrow{r}{Sf} & SY \arrow{d}{h_Y} \\
		X \arrow{r}[below]{f} &Y
		\end{tikzcd}.
	\end{equation*}
	 Well-definedness on morphisms follows from the naturality of $\alpha$, as seen on the right above.
	Compositionality follows immediately from the definition of $\alpha$ on morphisms.
\end{proof}

\begin{lemma}
\label[lemma]{alphapreservesfinal}
	Let $\alpha: \lambda^S \rightarrow \lambda^T$ be a distributive law homomorphism. If $\langle \Omega, h, k \rangle$ is the final $\lambda^S$-bialgebra, then $\langle \Omega, h \circ \alpha_{\Omega}, k \rangle$ is the final $\lambda^T$-bialgebra.
\end{lemma}
\begin{proof}
	It is well-known that if $\langle \Omega, h, k \rangle$ is the final $\lambda^S$-bialgebra, then $\langle \Omega, k \rangle$ is the final $F$-coalgebra and $h: S\Omega \rightarrow \Omega$ is the unique homomorphism satisfying $k \circ h = Fh \circ \lambda^S_{\Omega} \circ Sk$. Similarly, it is well-known that $\langle \Omega, \overline{h}, k \rangle$ is the final $\lambda^T$-bialgebra, where $\overline{h}: T\Omega \rightarrow \Omega$ is the unique homomorphism satisfying $k \circ \overline{h} = F\overline{h} \circ \lambda^T_{\Omega} \circ Tk$. The statement thus follows from uniqueness:
	\[
	\begin{tikzcd}
		T\Omega \arrow{d}[left]{Tk} \arrow{r}{\alpha_{\Omega}} & S\Omega \arrow{d}{Sk} \arrow{r}{h} & \Omega \arrow{dd}{k} \\
		TF\Omega \arrow{r}{\alpha_{F\Omega}} \arrow{d}[left]{\lambda^T_{\Omega}} & SF\Omega \arrow{d}{\lambda^S_{\Omega}} & \\
		FT\Omega \arrow{r}[below]{F\alpha_{\Omega}} & FS\Omega \arrow{r}[below]{Fh} & F\Omega
	\end{tikzcd}.
	\]
\end{proof}

\begin{oneshot}{corollary}{generatorbialgebrahom}
	Let $\alpha: \lambda^S \rightarrow \lambda^T$ be a homomorphism between distributive laws and $\langle X,h,k \rangle$ a $\lambda^S$-bialgebra. If $\langle Y, i, d \rangle$ is a generator for the $T$-algebra $\langle X, h \circ \alpha_X \rangle$, then $
	(h \circ \alpha_X) \circ Ti: \textnormal{exp}_T(\langle  Y, Fd \circ k \circ i \rangle) \rightarrow \langle X, h \circ \alpha_X, k \rangle
	$ is a $\lambda^T$-bialgebra homomorphism.
\end{oneshot}
\begin{proof}
	By \Cref{inducedbialgebra} the tuple $\langle X, h \circ \alpha_X, k \rangle$ constitutes a $\lambda^T$-bialgebra.
	The statement thus follows from \Cref{forgenerator-isharp-is-bialgebra-hom}.
\end{proof}

\begin{oneshot}{lemma}{distributivelawaxiomeasier}
	Let $\alpha: T \Rightarrow S$ be a natural transformation satisfying $h^{S} \circ \alpha_B = h^{T}$, then $\lambda^{S} \circ \alpha_F = F \alpha \circ \lambda^{T}$.
\end{oneshot}
\begin{proof}
	We need to establish the commutativity of the following diagram:
	\begin{equation*}
	\begin{tikzcd}[row sep=1.5em, column sep = 4em]
		T(X^A \times B) \arrow{r}{\alpha_{X^A \times B}} \arrow{d}[left]{\langle T \pi_1, T \pi_2 \rangle} & S(X^A \times B) \arrow{d}{\langle S \pi_1, S \pi_2 \rangle} \\
		T(X^A) \times TB \arrow{r}{\alpha_{X^A} \times \alpha_B} \arrow{d}[left]{\textnormal{st} \times h^{T}} & S(X^A) \times SB \arrow{d}{\textnormal{st} \times h^S}  \\
		(TX)^A \times B \arrow{r}[below]{(\alpha_X)^A \times B} & (SX)^A \times B
	\end{tikzcd}.
	\end{equation*}
The commutativity of the top square is a consequence of the naturality of $\alpha$. Similarly, the commutativity of the bottom square follows from the assumption and the naturality of $\alpha$,
\begin{align*}
	\textnormal{st} \circ \alpha_{X^A}(U)(a) 
	&= S(\textnormal{ev}_a) \circ \alpha_{X^A}(U) & \textnormal{(Def. st)} \\
	&= \alpha_{X} \circ T(\textnormal{ev}_a)(U) & \textnormal{(Nat. } \alpha) \\
	&= \alpha_X^A \circ \textnormal{st}(U)(a) & \textnormal{(Def. st)}.
\end{align*}
\end{proof}

\begin{oneshot}{proposition}{algebrainduceddistributivellawhom}
	Any algebra $h: T2 \rightarrow 2$ over a set monad $T$ induces a homomorphism $\alpha^{h}: \lambda^{\mathcal{H}} \rightarrow \lambda^{h}$ between distributive laws by $\alpha^{h}_X := h^{2^X} \circ \textnormal{st} \circ T(\eta^{\mathcal{H}}_X)$.\end{oneshot}
\begin{proof}
It is well-known that the strength operation is natural and satisfies the equalities
$
	(\eta_A^T)^B = \textnormal{st} \circ \eta^T_{A^B}$ and
	$\textnormal{st} \circ \mu_{A^B} = \mu_{A}^B \circ \textnormal{st} \circ T\textnormal{st}$.
It is also not hard to see that for functions $f: A \rightarrow B$ and $g: C \rightarrow D$ it holds $f^D \circ A^{g} = B^g \circ f^C$. We write $f^*$ for $A^f$ and $f_*$ for $f^A$, and omit components of natural transformations for readability. 
The naturality of $\alpha^{h^{T}}$ is a consequence of:
\[
			\begin{tikzcd}
				TX \arrow{d}[left]{Tf} \arrow{r}{T\eta^{\mathcal{H}}} & T(2^{2^X}) \arrow{d}{T((f^*)^*)} \arrow{r}{\textnormal{st}} & T(2)^{2^X} \arrow{r}{h_*} \arrow{d}{(f^*)^*}  & 2^{2^X} \arrow{d}{(f^*)^*} \\
				TY \arrow{r}[below]{T\eta^{\mathcal{H}}} & T(2^{2^Y}) \arrow{r}[below]{\textnormal{st}}  & T(2)^{2^Y} \arrow{r}[below]{h_*} & 2^{2^Y}
			\end{tikzcd}.
			\]
	Using the equality $2^{\eta^{\mathcal{H}}_{2^X}} \circ \eta^{\mathcal{H}}_{2^{2^X}} = \textnormal{id}_{2^{2^X}}$, the equation involving the monad multiplications is seen from:
		\begin{equation*}
			\begin{tikzcd}
				 T(2^{2^X}) \arrow{r}{T(\eta^{\mathcal{H}})}  \arrow{dr}{1} & T(2^{2^{2^{2^X}}}) \arrow{r}{\textnormal{st}} \arrow{d}{T((\eta^{\mathcal{H}})^*)} & T(2)^{2^{2^{2^X}}} \arrow{r}{h_*} \arrow{dd}{(\eta^{\mathcal{H}})^*} & 2^{2^{2^{2^X}}} \arrow{ddddd}[right]{(\eta^{\mathcal{H}})^*} \\
				& T(2^{2^X}) \arrow{dr}{\textnormal{st}} & & \\
				T(T(2)^{2^X}) \arrow{uu}[left]{T(h_{*})} \arrow{r}[below]{\textnormal{st}} & T^2(2)^{2^X} \arrow{dddr}{\mu^T_*} \arrow{r}[below]{T(h)_*} & T(2)^{2^X} \arrow{dddr}{h_*} & \\
				T^2(2^{2^X}) \arrow{ddr}{\mu^T} \arrow{u}[left]{T(\textnormal{st})} & & & \\
				T^2(X) \arrow{u}[left]{T^2(\eta^{\mathcal{H}})} \arrow{d}[left]{\mu^T} & & & \\
				T(X) \arrow{r}[below]{T(\eta^{\mathcal{H}})} & T(2^{2^X}) \arrow{r}[below]{\textnormal{st}} & T(2)^{2^X} \arrow{r}[below]{h_*} & 2^{2^X} 
			\end{tikzcd}.
		\end{equation*}
	The equation involving the monad units is established by:
	\begin{equation*}
			\begin{tikzcd}[row sep=2em, column sep = 2em]
				X \arrow{rrr}{\eta^T} \arrow{ddd}[left]{\eta^{\mathcal{H}}} \arrow{dr}{{\eta^{\mathcal{H}}}} & & & TX \arrow{dl}{T\eta^{\mathcal{H}}} \\
				& 2^{2^X} \arrow{d}{\eta^{T}_*} \arrow{r}{\eta^T} \arrow{ddl}[left]{1} & T(2^{2^X}) \arrow{dl}{\textnormal{st}} & \\
				& T(2)^{2^X} \arrow{dl}{h_*} & & \\
				2^{2^{X}} & & & 
			\end{tikzcd}.
	\end{equation*}
	To show that the equation involving the distributive laws holds, we use \Cref{distributivelawaxiomeasier}.  That is, we note that for any $f$ it holds $f \circ \textnormal{ev}_a = \textnormal{ev}_a \circ f_*$, and moreover, $h^{\mathcal{H}} = \textnormal{ev}_{\textnormal{id}_2}$, before establishing:
	\begin{equation*}
					\begin{tikzcd}[row sep=2em, column sep = 5em]
				T(2) \arrow{ddr}[left]{h} \arrow{r}{T(\eta^{\mathcal{H}})} \arrow{dr}{1} & T(2^{2^{2}}) \arrow{d}{T(\textnormal{ev}_{\textnormal{id}_2})} \arrow{r}{\textnormal{st}} & T(2)^{2^{2}} \arrow{r}{h_*} \arrow{dl}{\textnormal{ev}_{\textnormal{id}_2}} & 2^{2^2} \arrow{ddll}{h^{\mathcal{H}}} \\
				& T(2) \arrow{d}{h} & & \\
				& 2 & & 		
					\end{tikzcd}.
	\end{equation*}
\end{proof}

\begin{oneshot}{corollary}{alphapowersetneighbourhooddistrlaw}
	Let $\alpha_X: \mathcal{P}X \rightarrow \mathcal{H}X$ satisfy $\alpha_X(\varphi)(\psi) = \bigvee_{x \in X} \varphi(x) \wedge \psi(x)$, then $\alpha$ constitutes a distributive law homomorphism $\alpha: \lambda^{\mathcal{H}} \rightarrow \lambda^{\mathcal{P}}$.
\end{oneshot}
\begin{proof}
	We show that $\alpha^{h^{\mathcal{P}}} = \alpha$, the statement then follows from \Cref{algebrainduceddistributivellawhom}. We calculate
	\begin{align*}
		\alpha^{h^{\mathcal{P}}}_X(\varphi)(\psi) &= (h^{\mathcal{P}})^{2^X} \circ \textnormal{st}  \circ \mathcal{P}(\eta^{\mathcal{H}}_X)(\varphi)(\psi) & \textnormal{(Def. } \alpha^{h^{\mathcal{P}}}_X)  \\
		&= \textnormal{st} \circ \mathcal{P}(\eta^{\mathcal{H}}_X)(\varphi)(\psi)(1) & \textnormal{(Def. } h^{\mathcal{P}}) \\
		&= \mathcal{P}(\textnormal{ev}_{\psi})(\mathcal{P}(\eta^{\mathcal{H}}_X)(\varphi))(1) & \textnormal{(Def. st)} \\
		&= \mathcal{P}(\textnormal{ev}_{\psi} \circ \eta^{\mathcal{H}}_X)(\varphi)(1) & (\mathcal{P}(f) \circ \mathcal{P}(g) = \mathcal{P}(f \circ g)) \\
		&= \mathcal{P}(\psi)(\varphi)(1) & (\textnormal{Def. ev, }\eta^{\mathcal{H}}_X) \\
		&= \bigvee_{x \in \psi^{-1}(1)} \varphi(x) & (\textnormal{Def. } \mathcal{P}(\psi)) \\
		&= \bigvee_{x \in X} \varphi(x) \wedge \psi(x) & (x \in \psi^{-1}(1))  \\
		&= \alpha_X(\varphi)(\psi) & \textnormal{(Def. } \alpha_X) .
	\end{align*}
\end{proof}

\begin{oneshot}{lemma}{basisshiftecabapowerset}
	Let $\alpha_X: \mathcal{P}X \rightarrow \mathcal{H}X$ satisfy $\alpha_X(\varphi)(\psi) = \bigvee_{x \in X} \varphi(x) \wedge \psi(x)$. If $B = \langle X, h \rangle$ is a $\mathcal{H}$-algebra, then $\langle \textnormal{At}(B), i, d \rangle$ with $i(a) = a$ and $d(x) = \lbrace a \in \textnormal{At}(B) \mid a \leq x \rbrace$ is a basis for the $\mathcal{P}$-algebra $\langle X, h \circ \alpha_X \rangle$.
\end{oneshot}
\begin{proof}
	Let $K: \textnormal{Set}^{\textnormal{op}} \rightarrow \textnormal{Alg}(\mathcal{H})$ denote the comparison functor with $K(X) = \langle \mathcal{P}X, 2^{\eta^{\mathcal{H}}_X} \rangle$ induced by the self-dual contravariant powerset adjunction. It is well-known that $K$ has a quasi-inverse, namely the functor $\textnormal{At}: \textnormal{Alg}(\mathcal{H}) \rightarrow \textnormal{Set}^{\textnormal{op}}$ assigning to a complete atomic Boolean algebra $B$ its atoms $\textnormal{At}(B)$ \cite{taylor2002subspaces}. The equivalence $d: B \simeq K \circ \textnormal{At}(B)$ is given by $d(x) = \lbrace a \in \textnormal{At}(B) \mid a \leq x \rbrace$.
	The calculation below 
	\begin{align*}
		2^{\eta^{\mathcal{H}}_X} \circ \alpha_{\mathcal{P}X}(\Phi)(x) &=
		\alpha_{\mathcal{P}X}(\Phi)(\eta^{\mathcal{H}}_X(x)) & \textnormal{(Def. } 2^{\eta^{\mathcal{H}}_X}) \\
		&= \bigvee_{\varphi \in 2^X} \Phi(\varphi) \wedge \eta^{\mathcal{H}}_X(x)(\varphi) &\textnormal{(Def. } \alpha_{\mathcal{P}X}) \\
		&= \bigvee_{\varphi \in 2^X} \Phi(\varphi) \wedge \varphi(x) & \textnormal{(Def. } \eta_X^{\mathcal{H}} ) \\
		&= \mu_{X}^{\mathcal{P}}(\Phi)(x) & \textnormal{(Def. } \mu^{\mathcal{P}}_X)
	\end{align*}
	shows that $2^{\eta^{\mathcal{H}}_X} \circ \alpha_{\mathcal{P}X} = \mu^{\mathcal{P}}_X$. By \Cref{alphapowersetneighbourhooddistrlaw} the definition $\alpha(\langle X, h \rangle) = \langle X, h \circ \alpha_X \rangle$ yields a functor $\alpha: \textnormal{Alg}(\mathcal{H}) \rightarrow \textnormal{Alg}(\mathcal{P})$. 
	We can thus deduce the following equivalence of $\mathcal{P}$-algebras
	\begin{align*}
		\langle X, h \circ \alpha_X \rangle &= \alpha(B) & \textnormal{(Def. } \alpha) \\
		&\simeq \alpha \circ K \circ \textnormal{At}(B) & (\textnormal{id} \simeq K \circ \textnormal{At}) \\
		&= \langle \mathcal{P}(\textnormal{At}(B)), 2^{\eta^{\mathcal{H}}_{\textnormal{At}(B)}} \circ \alpha_{\mathcal{P}(\textnormal{At}(B))}  \rangle & \textnormal{(Def. } \alpha \circ K \circ \textnormal{At}) \\
		&=  \langle \mathcal{P}(\textnormal{At}(B)), \mu^{\mathcal{P}}_{\textnormal{At}(B)} \rangle & (2^{\eta^{\mathcal{H}}_X} \circ \alpha_{\mathcal{P}X} = \mu^{\mathcal{P}}_X).
	\end{align*}
	Using the definition of a basis, the former immediately implies the claim.
	\end{proof}

\begin{oneshot}{corollary}{neighbourhoodpowersetmorphism}
	Let $\alpha_X: \mathcal{P}X \rightarrow \mathcal{A}X$ satisfy $\alpha_X(\varphi)(\psi) = \bigvee_{x \in X} \varphi(x) \wedge \psi(x)$, then $\alpha$ constitutes a distributive law homomorphism $\alpha: \lambda^{\mathcal{A}} \rightarrow \lambda^{\mathcal{P}}$.
\end{oneshot}
\begin{proof}
	We observe that $\alpha_X(\varphi): \langle 2^X, \subseteq \rangle \rightarrow \langle 2, \leq \rangle$ is monotone for all $\varphi \in 2^X$. Since the monotone neighbourhood monad $\mathcal{A}$ and the neighbourhood monad $\mathcal{H}$ only differ on objects, the result follows from \Cref{alphapowersetneighbourhooddistrlaw}.  
\end{proof}

\begin{oneshot}{corollary}{alphaxorneighbourhooddistrlaw}
	Let $\alpha_X: \mathcal{R}X \rightarrow \mathcal{H}X$ satisfy $\alpha_X(\varphi)(\psi) = \bigoplus_{x \in X} \varphi(x) \cdot \psi(x)$, then $\alpha$ constitutes a distributive law homomorphism $\alpha: \lambda^{\mathcal{H}} \rightarrow \lambda^{\mathcal{R}}$.
\end{oneshot}
\begin{proof}
	Analogous to the proof of \Cref{alphapowersetneighbourhooddistrlaw}.
\end{proof}

\begin{oneshot}{proposition}{alphaunderlies}
Let $\alpha: \lambda^{S} \rightarrow \lambda^T$ be a distributive law homomorphism. Then $\alpha_X: TX \rightarrow SX$ underlies a natural transformation $\alpha: \textnormal{exp}_T \Rightarrow \alpha \circ \textnormal{exp}_S \circ \textnormal{ext}$ between functors of type $\textnormal{Coalg}(FT) \rightarrow \textnormal{Bialg}(\lambda^T)$.
\end{oneshot}
\begin{proof}
	Given a $T$-succinct automaton $\mathcal{X} = \langle X, k \rangle$ the definitions imply 
	\begin{align*}
		\textnormal{exp}_T(\mathcal{X}) &= \langle TX, \mu^T_X, F \mu^T_X \circ \lambda^T_{TX} \circ Tk \rangle \\
		\alpha \circ \textnormal{exp}_S \circ \textnormal{ext}(\mathcal{X}) &= \langle SX, \mu^S_X \circ \alpha_{SX}, F\mu^S_X \circ \lambda^S_{SX} \circ SF\alpha_X \circ Sk \rangle.
	\end{align*}
	By the definition of distributive law homomorphisms, the morphism $\alpha_X$ commutes with the underlying $T$-algebra structures. Its commutativity with the underlying $F$-coalgebra structures follows from:
	\[
	\begin{tikzcd}
		TX \arrow{dd}[left]{\alpha_X} \arrow{r}{Tk} & TFTX \arrow{dd}{\alpha_{FTX}} \arrow{rr}{\lambda^T_{TX}} \arrow{dr}{TF\alpha_X} & & FT^2X \arrow{r}{F\mu^T_X} \arrow{d}{FT\alpha_X} & FTX \arrow{dd}{F\alpha_X} \\
		&& TFSX \arrow{r}{\lambda^T_{SX}} \arrow{d}{\alpha_{FSX}} & FTSX \arrow{d}{F\alpha_{SX}} &  \\
		SX \arrow{r}{Sk} & SFTX \arrow{r}{SF \alpha_X} & SFSX \arrow{r}{\lambda^S_{SX}} & FS^2X \arrow{r}{F\mu^S_X} & FSX
	\end{tikzcd}.
	\]
	For above we use the naturality of $\alpha$ and $\lambda^T$, and the definition of a distributivity law homomorphism. The naturality of $\alpha$ as natural transformation $\alpha: \textnormal{exp}_T \Rightarrow \alpha \circ \textnormal{exp}_S \circ \textnormal{ext}$ follows immediately from the naturality of $\alpha$ as natural transformation $\alpha: T \Rightarrow S$.
\end{proof}

\begin{oneshot}{lemma}{generatorclosed}
	Let $\alpha: \lambda^{S} \rightarrow \lambda^T$ be a distributive law homomorphism and $\langle X, h, k \rangle$ a $\lambda^S$-bialgebra.
	If $\langle Y, i, d \rangle$ is a generator for $\langle X, h \circ \alpha_X \rangle$, then $\langle Y,  Fd \circ k \circ i\rangle$ is $\alpha$-closed.
\end{oneshot}
\begin{proof}
	We write $\mathbb{X} := \langle X, h, k \rangle$, $\mathbb{G} := \langle Y, i, d \rangle$, and $\gen(\alpha(\mathbb{X}), \mathbb{G}) := \langle Y,Fd \circ k \circ i\rangle$.  The definitions imply
	\begin{align*}
		\textnormal{exp}_T(\gen(\alpha(\mathbb{X}), \mathbb{G})) &= \langle TY, \mu^T_Y, (Fd \circ k \circ i)^{\sharp} \rangle \\
		\alpha \circ \textnormal{exp}_S \circ \textnormal{ext}(\gen(\alpha(\mathbb{X}), \mathbb{G})) &= \langle SY, \mu^S_Y \circ \alpha_{SY}, (F(\alpha_Y \circ d) \circ k \circ i)^{\sharp} \rangle.
	\end{align*}
	Since $\mathbb{G}$ is a generator for $\langle X, h \circ \alpha_X \rangle$, \Cref{forgenerator-isharp-is-bialgebra-hom} implies that $(h \circ \alpha_X) \circ Ti: \textnormal{exp}_T(\gen(\alpha(\mathbb{X}), \mathbb{G})) \rightarrow \alpha(\mathbb{X})$ is a surjective $\lambda^T$-bialgebra homomorphism. 
	Naturality of $\alpha$ shows that $\overline{G} = \langle Y, i, \alpha_Y \circ d \rangle$ is a generator for $\langle X, h \rangle$. Thus \Cref{forgenerator-isharp-is-bialgebra-hom} implies that $h \circ Si: \textnormal{exp}_S(\gen(\mathbb{X}, \overline{\mathbb{G}})) \rightarrow \mathbb{X}$ is a surjective $\lambda^S$-bialgebra homomorphism. Applying $\alpha$ to the former shows that $h \circ Si: \alpha \circ \textnormal{exp}_S \circ \textnormal{ext}(\gen(\alpha(\mathbb{X}), \mathbb{G})) \rightarrow \alpha(\mathbb{X})$ is a surjective $\lambda^T$-bialgebra homomorphism.
	The statement follows from the uniqueness of epi-mono factorisations:
	\[
	\begin{tikzcd}
		\textnormal{exp}_T(\gen(\alpha(\mathbb{X}), \mathbb{G})) \arrow[twoheadrightarrow]{rrr}{\obs} \arrow{d}[left]{\alpha_Y} \arrow[twoheadrightarrow]{dr}{(h \circ \alpha_X) \circ Ti} & & & \img(\obs_{\textnormal{exp}_T(\gen(\alpha(\mathbb{X}), \mathbb{G}))}) \arrow[hookrightarrow]{dd}{} \arrow[dashed]{dl}{\simeq} \\
		\alpha \circ \textnormal{exp}_S \circ \textnormal{ext}(\gen(\alpha(\mathbb{X}), \mathbb{G})) \arrow[twoheadrightarrow]{d}[left]{\obs}\arrow[twoheadrightarrow]{r}{h \circ Si} & \alpha(\mathbb{X}) \arrow[twoheadrightarrow]{r}{\obs} &\img(\obs_{\alpha(\mathbb{X})}) \arrow[dashed]{dll}{\simeq} \arrow[hookrightarrow]{dr}{} & \\
		\img(\obs_{\alpha \circ \textnormal{exp}_S \circ \textnormal{ext}(\gen(\alpha(\mathbb{X}), \mathbb{G}))}) \arrow[hookrightarrow]{rrr}{} & & & \Omega
	\end{tikzcd}.
	\]
\end{proof}

\begin{oneshot}{theorem}{minimalitytheorem}
	Given a language $\mathcal{L} \in \Omega$ such that there exists a minimal pointed $\lambda^S$-bialgebra $\mathbb{M}$ accepting $\mathcal{L}$ and the underlying algebra of $\alpha(\mathbb{M})$ admits a size-minimal generator, there exists a pointed $\alpha$-closed $T$-succinct automaton $\mathcal{X}$ accepting $\mathcal{L}$ such that: \begin{itemize}
		\item for any pointed $\alpha$-closed $T$-succinct automaton $\mathcal{Y}$ accepting $\mathcal{L}$ we have that $\img(\obs_{\expa_T(\mathcal{X})}) \subseteq \img(\obs_{\expa_T(\mathcal{Y})})$;
		\item if $\img(\obs_{\expa_T(\mathcal{X})}) = \img(\obs_{\expa_T(\mathcal{Y})})$, then $\vert X \vert \leq \vert Y \vert$, where $X$ and $Y$ are the carriers of $\mathcal{X}$ and $\mathcal{Y}$, respectively.
	\end{itemize} 
 \end{oneshot}
\begin{proof}
We use a similar notation as in the proof of \Cref{generatorclosed}.
	Let $\mathbb{G} = \langle X, i, d \rangle$ be the size-minimal generator for the underlying algebra of $\alpha(\mathbb{M})$, which we assume to be $x$-pointed.
	We define a $d(x)$-pointed $T$-succinct automaton $\mathcal{X} := \gen(\alpha(\mathbb{M}), \mathbb{G})$ , which by the existence of the pointed $\lambda^T$-algebra homomorphism $i^\sharp \colon \expa(\mathcal{X}) \to \alpha(\mathbb{M})$ accepts the language accepted by $\alpha(\mathbb{M})$.
	Because $\alpha$ only modifies the algebraic part of a bialgebra and the final bialgebra homomorphism is induced by the underlying final coalgebra homomorphism, the language accepted by $\alpha(\mathbb{M})$ is the language $\mathcal{L}$ accepted by $\mathbb{M}$.
	From \Cref{generatorclosed} it follows that $\mathcal{X}$ is $\alpha$-closed.

	Consider any pointed $\alpha$-closed $T$-succinct automaton $\mathcal{Y}$ accepting $\mathcal{L}$.
	Then by minimality of $\mathbb{M}$ there exists an injective $\lambda^S$-bialgebra homomorphism $j \colon \mathbb{M} \to \img(\obs_{\expa_S(\ext(\mathcal{Y}))})$, which is also a $\lambda^T$-bialgebra homomorphism $j \colon \alpha(\mathbb{M}) \to \img(\obs_{\alpha(\expa_S(\ext(\mathcal{Y}))}))$, because the functor $\alpha$ is the identity on morphisms, and only modifies the algebraic part of a bialgebra.
	Since $\mathcal{Y}$ is $\alpha$-closed, the codomain of $j$ is isomorphic to $\img(\obs_{\expa_T(\mathcal{Y})})$.
	Moreover, $\alpha(\mathbb{M})$ is isomorphic to $\img(\obs_{\expa_T(\mathcal{X})})$ by the existence of a surjective homomorphism $\expa_T(\mathcal{X}) \to \alpha(\mathbb{M})$.
	The resulting $\lambda^T$-bialgebra homomorphism $\img(\obs_{\expa_T(\mathcal{X})}) \to \img(\obs_{\expa_T(\mathcal{Y})})$ commutes with observability maps and thus must be an inclusion map, so $\img(\obs_{\expa_T(\mathcal{X})}) \subseteq \img(\obs_{\expa_T(\mathcal{Y})})$.

	Suppose $\img(\obs_{\expa_T(\mathcal{X})}) = \img(\obs_{\expa_T(\mathcal{Y})})$, which implies that $j$ is an isomorphism.
	Then there exists a surjective $\lambda^T$-bialgebra homomorphism $\expa_T(\mathcal{Y}) \to \alpha(\mathbb{M})$, which means that $Y$ forms the carrier of a generator for the underlying algebra of $\alpha(\mathbb{M})$.
	By the size-minimality of $\mathbb{G}$ we thus obtain $\vert X \vert \leq \vert Y \vert$.
\end{proof}

\begin{lemma}
\label[lemma]{xorbasisstateminimal}
	Any basis for a $\mathcal{R}$-algebra is a size-minimal generator.
\end{lemma}
\begin{proof}
	Let $\langle Y, i, d \rangle$ be a basis for a $\mathcal{R}$-algebra $\langle X, h \rangle$, then $d: X \rightarrow \mathcal{R}Y$ is a bijection with inverse $h \circ \mathcal{R}i$. Let $\langle Y', i', d' \rangle$ be any other generator for $\langle X, h \rangle$. Then $h \circ \mathcal{R}i': \mathcal{R}Y' \rightarrow X$ is a surjection, which shows that $d \circ h \circ \mathcal{R}i': \mathcal{R}Y' \rightarrow \mathcal{R}Y$ is a surjection. In consequence, $\vert \mathcal{R}Y \vert \leq \vert \mathcal{R} Y' \vert$, which implies $\vert Y \vert \leq \vert Y' \vert$, since $\mathcal{R}- = 2^{-}$.
\end{proof}

\begin{lemma}
\label[lemma]{basisstateminimalpowerset}
	Any basis for a $\mathcal{P}$-algebra is a size-minimal generator.
\end{lemma}
\begin{proof}
	Analogous to the proof of \Cref{xorbasisstateminimal}. 
\end{proof}

\begin{corollary}
\label[corollary]{atomsstateminimalgenerator}
	Let $\alpha_X: \mathcal{P}X \rightarrow \mathcal{H}X$ satisfy $\alpha_X(\varphi)(\psi) = \bigvee_{x \in X} \varphi(x) \wedge \psi(x)$. If $B = \langle X, h \rangle$ is a $\mathcal{H}$-algebra, then $\langle \textnormal{At}(B), i, d \rangle$ with $i(a) = a$ and $d(x) = \lbrace a \in \textnormal{At}(B) \mid a \leq x \rbrace$ is a size-minimal generator for $\langle X, h \circ \alpha_X \rangle$. 
	\end{corollary}
\begin{proof}
	By \Cref{basisshiftecabapowerset} $\langle \textnormal{At}(B), i, d \rangle$ is a basis for $\langle X, h \circ \alpha_X \rangle$, which by \Cref{basisstateminimalpowerset} implies size-minimality.
\end{proof}

\begin{lemma}
\label[lemma]{joinirreducstateminimal}
	For any finite $\mathcal{P}$-algebra $L = \langle X, h \rangle$ the join-irreducibles $\langle J(L), i, d \rangle$ with $i(y) = y$ and $d(x) = \lbrace y \in J(L) \mid y \leq x \rbrace$ constitute a size-minimal generator.
	 \end{lemma}
\begin{proof}
	Since $L$ is finite, it satisfies the descending chain condition (DCC), which in turn can be used to show that the join-irreducibles constitute a generator as follows.

Assume there exists some $x \in X$ with $x \not= i^{\sharp}(d(x))$. We build an infinite sequence $(a_n)$ with $a_i > a_{i+1}$ and $a_i \not = i^{\sharp}(d(a_i))$, which contradicts the DCC. For the base case we define $a_0 := x$. For any $x \in X$, the property $x \in J(L)$ immediately implies $x = i^{\sharp}(d(x))$. Thus we can assume $a_i \not \in J(L)$.
In consequence we have $a_i = y \vee z$ for $y,z \not= a_i$, i.e $a_i > y$ and $a_i > z$. 
Assume $y = i^{\sharp}(d(y))$ and $z = i^{\sharp}(d(z))$. 
Then 
\[ i^{\sharp}(d(a_i)) \leq a_i = y \vee z = i^{\sharp}(d(y)) \vee i^{\sharp}(d(z)) = i^{\sharp}(d(y) \vee d(z)) \leq i^{\sharp}(d(a_i)). \]
It thus follows $a_i = i^{\sharp}(d(a_i))$, which is a contradiction. Hence, w.l.o.g. assume $y \not = i^{\sharp}(d(y))$ and define $a_{i+1} := y$.
	
	Let $\langle Y, i', d' \rangle$ be an arbitrary generator for $L$. For any $a \in J(L)$ we have $a = \bigvee^h_{y \in d'(a)} i'(y)$. By the definition of join-irreducibles there exists at least one $y_a \in d'(a)$ such that $i'(y_a) = a$. One can thus define a function $f: J(L) \rightarrow Y$ with $f(a) = y_a$. It is not hard to see that $f$ is injective, which implies $\vert J(L) \vert \leq \vert Y \vert$.
\end{proof}

\begin{lemma}
\label[lemma]{obsgenerated}
	Let $\mathbb{X} = \langle X, h, k \rangle$ be an observable $\lambda^S$-bialgebra and $\mathbb{G}$ a generator for $\langle X, h \circ \alpha_X \rangle$, then $\img(\obs_{\textnormal{exp}_T(\textnormal{gen}(\alpha(\mathbb{X}), \mathbb{G}))}) \simeq X$.
\end{lemma}
\begin{proof}
By \Cref{forgenerator-isharp-is-bialgebra-hom} there exists a surjective $\lambda^T$-bialgebra homomorphism $\textnormal{exp}_T(\textnormal{gen}(\alpha(\mathbb{X}), \mathbb{G})) \rightarrow \alpha(\mathbb{X})$. Since the final $\lambda^T$-bialgebra homomorphism is induced by the underlying final $F$-coalgebra homomorphism and  $\alpha(\mathbb{X}) = \langle X, h \circ \alpha_X, k \rangle$, it holds $\obs_{\alpha(\mathbb{X})} = \obs_{\mathbb{X}}$. The statement follows from the uniqueness of epi-mono factorizations and the definition of $\alpha(\mathbb{X})$:
	\[
	\begin{tikzcd}
		\textnormal{exp}_T(\textnormal{gen}(\alpha(\mathbb{X}), \mathbb{G})) \arrow[twoheadrightarrow]{r}{} \arrow[twoheadrightarrow]{d}{}  & \alpha(\mathbb{X}) \arrow[hookrightarrow]{d}{\obs_{\alpha(\mathbb{X})} = \obs_{\mathbb{X}}} \arrow[dashed]{dl}{\simeq} \\
		\img(\obs_{\textnormal{exp}_T(\textnormal{gen}(\alpha(\mathbb{X}), \mathbb{G}))}) \arrow[hookrightarrow]{r}{} & \Omega
	\end{tikzcd}.
	\]
\end{proof}

\begin{lemma}
\label[lemma]{obsdagext}
Let $\mathcal{X}$ be a $T$-succinct automaton, then $\obs^{\dag}_{\mathcal{X}} = \obs^{\dag}_{\ext(\mathcal{X})}$.
\end{lemma}
\begin{proof}
	Since by \Cref{alphaunderlies} the morphism $\alpha_X: \textnormal{exp}_T(\mathcal{X}) \rightarrow \alpha(\textnormal{exp}_S(\textnormal{ext}(\mathcal{X})))$ is a $\lambda^T$-bialgebra homomorphism, we have by uniqueness $\obs_{\alpha(\exp_S(\ext(\mathcal{X})))} \circ \alpha_X = \obs_{\exp_T(\mathcal{X})}$. Since any final bialgebra homomorphism is induced by the underlying final $F$-coalgebra homomorphism it holds $\obs_{\alpha(\exp_S(\ext(\mathcal{X})))} = \obs_{\exp(\ext_S(\mathcal{X}))}$, which thus implies 
	\begin{equation}
	\label{obsdagproofeq}
		\obs_{\exp_S(\ext(\mathcal{X}))} \circ \alpha_X = \obs_{\exp_T(\mathcal{X})}.
	\end{equation}
	 The statement follows from
	\begin{align*}
		\obs^{\dag}_{\mathcal{X}} &= \obs_{\exp_T(\mathcal{X})} \circ \eta^T_X & \textnormal{(Def. } \obs^{\dag}_{\mathcal{X}}) \\
		&= \obs_{\exp_S(\ext(\mathcal{X}))} \circ \alpha_X \circ \eta^T_X & \eqref{obsdagproofeq} \\
		&= \obs_{\exp_S(\ext(\mathcal{X}))} \circ \eta^S_X & \textnormal{(Def. distr. law hom.)}\\
		&= \obs^{\dag}_{\ext(\mathcal{X})} & \textnormal{(Def. } \obs^{\dag}_{\ext(\mathcal{X})}).
	\end{align*}
\end{proof}

\begin{oneshot}{lemma}{imgobssuccinctsem}
Let $\alpha: \lambda^{S} \rightarrow \lambda^T$ be a distributive law homomorphism. For any $T$-succinct automaton $\mathcal{X}$ it holds that $\img(\obs_{\exp_T(\mathcal{X})}) = \img(h \circ \alpha_{\Omega} \circ T(\obs^{\dag}_{\mathcal{X}}))$ and 
		$\img(\obs_{\alpha(\exp_S(\ext(\mathcal{X})))}) = \img(h \circ S(\obs^{\dag}_{\mathcal{X}}))$, where $\langle \Omega, h, k \rangle$ is the final $\lambda^S$-bialgebra.
	\end{oneshot}
\begin{proof}
	By \Cref{alphapreservesfinal} $\langle \Omega, h \circ \alpha_{\Omega}, k \rangle$ is the final $\lambda^T$-bialgebra. The first statement follows from
	\begin{align*}
		\obs_{\exp_T(\mathcal{X})} 
		&= \obs_{\exp_T(\mathcal{X})} \circ \mu^T_X  \circ T(\eta^T_X) & \textnormal{(Def. monad } T) \\
		&= h \circ \alpha_{\Omega} \circ T(\obs_{\exp_T(\mathcal{X})}) \circ T(\eta^T_X) & \textnormal{(Algebra hom. } \obs_{\exp_T(\mathcal{X})}) \\
		&= h \circ \alpha_{\Omega} \circ T(\obs^{\dag}_{\mathcal{X}}) & \textnormal{(Def. } \obs^{\dag}_{\mathcal{X}}).
	\end{align*}
	Similarly one shows that $	\obs_{\exp_S(\ext(\mathcal{X}))} =  h \circ S(\obs^{\dag}_{\ext(\mathcal{X})})$.
 Since any final bialgebra homomorphism is induced by the underlying final $F$-coalgebra homomorphism, it thus follows
	\begin{align*}
		\img(\obs_{\alpha(\exp_S(\ext(\mathcal{X})))}) &= \img(\obs_{\exp_S(\ext(\mathcal{X}))}) & (\obs_{\alpha(\exp_S(\ext(\mathcal{X})))} = \obs_{\exp_S(\ext(\mathcal{X}))}) \\
		&= \img(h \circ S(\obs^{\dag}_{\ext(\mathcal{X})})) & (\obs_{\exp_S(\ext(\mathcal{X}))} =  h \circ S(\obs^{\dag}_{\ext(\mathcal{X})})) \\
		&= \img(h \circ S(\obs^{\dag}_{\mathcal{X}})) & \textnormal{(\Cref{obsdagext})}.
	\end{align*}
\end{proof}

\begin{corollary}
\label[corollary]{imgcabanfa}
	Let $\alpha: \mathcal{P}X \rightarrow \mathcal{H}X$ satisfy $\alpha_X(\varphi)(\psi) = \bigvee_{x \in X} \varphi(x) \wedge \psi(x)$. For any unpointed non-deterministic automaton $\mathcal{X}$ it holds:
	\begin{itemize}
		\item $\img(\obs_{\exp_{\mathcal{P}}(\mathcal{X})}) = \overline{\textnormal{im}(\obs^{\dag}_{\mathcal{X}})}^{\textnormal{CSL}}$;
		\item $\img(\obs_{\alpha(\exp_{\mathcal{H}}(\ext(\mathcal{X})))}) = \overline{\textnormal{im}(\obs^{\dag}_{\mathcal{X}})}^{\textnormal{CABA}}$.
	\end{itemize}
\end{corollary}
\begin{proof}
	The final $\lambda^{\mathcal{H}}$-bialgebra is given by $\langle 2^{A^*}, 2^{\eta^{\mathcal{H}}_{A^*}}, \langle \varepsilon, \delta \rangle \rangle$. In the proof of \Cref{basisshiftecabapowerset} it was shown that $2^{\eta^{\mathcal{H}}_X} \circ \alpha_{\mathcal{P}X} = \mu^{\mathcal{P}}_X$. Thus it follows
	\begin{align*}
		\img(\obs_{\exp_{\mathcal{P}}(\mathcal{X})}) &= \img(2^{\eta^{\mathcal{H}}_{A^*}} \circ \alpha_{2^{A^*}} \circ \mathcal{P}(\obs^{\dag}_{\mathcal{X}})) & \textnormal{(\Cref{imgobssuccinctsem})} \\
		&= \img(\mu^{\mathcal{P}}_{A^*} \circ \mathcal{P}(\obs^{\dag}_{\mathcal{X}})) & (2^{\eta^{\mathcal{H}}_X} \circ \alpha_{\mathcal{P}X} = \mu^{\mathcal{P}}_X) \\
		&= \lbrace \bigcup_{u \in U} \obs^{\dag}_{\mathcal{X}}(u) \mid U \subseteq X \rbrace & \textnormal{(Def. } \mathcal{P}(-), \mu^{\mathcal{P}}) \\
		&= \overline{\lbrace \obs^{\dag}_{\mathcal{X}}(x) \mid x \in X \rbrace}^{\textnormal{CSL}} & \textnormal{(Def. } \overline{(-)}^{\textnormal{CSL}}).
	\end{align*}
	Similarly one computes
	\begin{align*}
\img(\obs_{\alpha(\exp_{\mathcal{H}}(\ext(\mathcal{X})))}) &= \img(2^{\eta^{\mathcal{H}}_{A^*}} \circ \mathcal{H}(\obs^{\dag}_{\mathcal{X}})) & \textnormal{(\Cref{imgobssuccinctsem})} \\
&= \lbrace \lbrace w \in A^* \mid \lbrace x \in X \mid \textnormal{obs}^{\dag}_{\mathcal{X}}(x)(w) = 1 \rbrace \in \Phi \rbrace \mid \Phi \subseteq 2^X \rbrace & \textnormal{(Def. } 2^{\eta^{\mathcal{H}}_{A^*}}, \mathcal{H}(-)) \\
		&= \lbrace \bigcup_{\varphi \in \Phi} \bigcap_{x \in \varphi} \obs^{\dag}_{\mathcal{X}}(x) \cap \bigcap_{x \not \in \varphi} \obs^{\dag}_{\mathcal{X}}(x)^c  \mid \Phi \subseteq 2^X \rbrace & \textnormal{(Set equality)}  \\
		&= \overline{\lbrace \obs^{\dag}_{\mathcal{X}}(x) \mid x \in X \rbrace}^{\textnormal{CABA}}  & \textnormal{(Def. } \overline{(-)}^{\textnormal{CABA}}).
	\end{align*}
\end{proof}

\begin{corollary}
\label[corollary]{imgdistromaton}
	Let $\alpha: \mathcal{P}X \rightarrow \mathcal{A}X$ satisfy $\alpha_X(\varphi)(\psi) = \bigvee_{x \in X} \varphi(x) \wedge \psi(x)$. For any unpointed non-deterministic automaton $\mathcal{X}$ it holds:
	\begin{itemize}
		\item $\img(\obs_{\exp_{\mathcal{P}}(\mathcal{X})}) = \overline{\textnormal{im}(\obs^{\dag}_{\mathcal{X}})}^{\textnormal{CSL}}$;
		\item $\img(\obs_{\alpha(\exp_{\mathcal{A}}(\ext(\mathcal{X})))}) = \overline{\textnormal{im}(\obs^{\dag}_{\mathcal{X}})}^{\textnormal{CDL}}$.
	\end{itemize}
\end{corollary}
\begin{proof}
	Analogous to the proof of \Cref{imgcabanfa}.
\end{proof}

\begin{corollary}
\label[corollary]{imgxorcaba}
	Let $\alpha: \mathcal{R}X \rightarrow \mathcal{H}X$ satisfy $\alpha_X(\varphi)(\psi) = \bigoplus_{x \in X} \varphi(x) \cdot \psi(x)$. For any unpointed $\mathbb{Z}_2$-weighted automaton $\mathcal{X}$ it holds:
	\begin{itemize}
		\item $\img(\obs_{\exp_{\mathcal{R}}(\mathcal{X})}) = \overline{\textnormal{im}(\obs^{\dag}_{\mathcal{X}})}^{\mathbb{Z}_2\textnormal{-Vect}}$;
		\item $\img(\obs_{\alpha(\exp_{\mathcal{H}}(\ext(\mathcal{X})))}) = \overline{\textnormal{im}(\obs^{\dag}_{\mathcal{X}})}^{\textnormal{CABA}}$.
	\end{itemize}
\end{corollary}
\begin{proof}
	Analogous to the proof of \Cref{imgcabanfa}.
\end{proof}

\begin{lemma}\tagcite{stackex}
\label[lemma]{join-irred-dlattice}
	Let $A$ be a sub-lattice of a finite distributive lattice $B$, then $\vert J(A) \vert \leq \vert J(B) \vert$.   
\end{lemma}
\begin{proof}
	For $x \in J(B)$ define $\hat{x} := \bigwedge \lbrace y \in A \mid x \leq y \rbrace \geq x$. To see that $\hat{x} \in J(A)$, assume $\hat{x} = y \vee z$ for $y,z \in A$. By distributivity we have	$
	x = \hat{x} \wedge x = (y \vee z) \wedge x = (y \wedge x ) \vee (z \wedge x)$.
	Since $x \in J(B)$, it thus follows w.l.o.g. $x = y \wedge x$, which implies $x \leq y$. Consequently $\hat{x} \leq y \leq \hat{x}$, i.e. $\hat{x} = y$. Let $z \in J(A)$, then the join-density of join-irreducibles implies
	\[
	z = \bigvee \lbrace x \in J(B) \mid x \leq z \rbrace = \bigvee \lbrace \hat{x} \in J(A) \mid x \in J(B) : x \leq z \rbrace .
	\]
	Since $z$ is join-irreducible it follows $z = \widehat{x_z}$ for some $x_z \in J(B)$ with $x_z \leq z$. We thus find
	$
	J(A) = \lbrace \hat{x} \mid x \in J(B) \rbrace$, which implies the claim $\vert J(A) \vert \leq \vert J(B) \vert$.
\end{proof}

\begin{corollary}
\label[corollary]{atomscaba}
	Let $A$ be a sub-algebra of a finite atomic Boolean algebra $B$. Then $\vert \textnormal{At}(A) \vert \leq  \vert \textnormal{At}(B) \vert$.
\end{corollary}
\begin{proof}
	For atomic Boolean algebras, join-irreducibles and atoms coincide. Every Boolean algebra is in particular a distributive lattice. The claim thus follows from \Cref{join-irred-dlattice}.
\end{proof}

\begin{oneshot}{corollary}{canonicalrfsaminimal}
	The canonical RFSA for $\mathcal{L}$ is size-minimal among non-deterministic automata $\mathcal{Y}$ accepting $\mathcal{L}$ with $\overline{\textnormal{im}(\textnormal{obs}^{\dagger}_{\mathcal{Y}})}^{\textnormal{CSL}} \subseteq \overline{\textnormal{Der}(\mathcal{L})}^{\textnormal{CSL}}$. 
\end{oneshot}
\begin{proof}
	By \Cref{powersetalgebra} the morphism $h^{\mathcal{P}}: \mathcal{P}2 \rightarrow 2$ with $h^{\mathcal{P}}(\varphi) = \varphi(1)$ is a $\mathcal{P}$-algebra. As shown in \Cref{induceddistrlaw}, it can used to derive a canonical distributive law $\lambda^{\mathcal{P}}$. It is not hard to see that the minimal pointed $\lambda^{\mathcal{P}}$-bialgebra $\mathbb{M}$ accepting $\mathcal{L}$ exists and that its underlying state space is given by the finite complete join-semi lattice $\overline{\textnormal{Der}(\mathcal{L})}^{\textnormal{CSL}}$. By \Cref{joinirreducstateminimal} the join-irreducibles for $\mathbb{M}$ constitute a size-minimal generator $\mathbb{G}$. By definition, the canonical RFSA for $\mathcal{L}$ is given by $\mathcal{X} := \textnormal{gen}(\mathbb{M}, \mathbb{G})$. From \Cref{obsgenerated} it follows that $\img(\obs_{\exp_{\mathcal{P}}(\mathcal{X})}) \simeq \overline{\textnormal{Der}(\mathcal{L})}^{\textnormal{CSL}}$. As seen in e.g. \Cref{imgcabanfa}, one has $\img(\obs_{\exp_{\mathcal{P}}(\mathcal{Y})}) =\overline{\textnormal{im}(\textnormal{obs}^{\dagger}_{\mathcal{Y}})}^{\textnormal{CSL}}$ for any NFA $\mathcal{Y}$. By choosing $\alpha$ as the identity, which implies $\alpha$-closedness for any NFA, the statement thus follows from \Cref{minimalitytheorem}.
\end{proof}

\begin{oneshot}{corollary}{minimalxor}
		The minimal xor automaton for $\mathcal{L}$ is size-minimal among $\mathbb{Z}_2$-weighted automata accepting $\mathcal{L}$.
\end{oneshot}
\begin{proof}
	Analogous to \Cref{canonicalrfsaminimal} one can show that the minimal xor automaton for $\mathcal{L} \subseteq A^*$ is size-minimal among all $\mathbb{Z}_2$-weighted automata $\mathcal{Y}$ accepting $\mathcal{L}$ such that
$ \overline{\textnormal{im}(\textnormal{obs}^{\dagger}_{\mathcal{Y}})}^{\mathbb{Z}_2\textnormal{-Vect}} \subseteq \overline{\textnormal{Der}(\mathcal{L})}^{\mathbb{Z}_2\textnormal{-Vect}}$.
	Specific to this case are \Cref{xoroutputalgebra} and \Cref{xorbasisstateminimal}. It remains to observe that for any $\mathbb{Z}_2$-weighted automaton $\mathcal{X}$, one can find an equivalent $\mathbb{Z}_2$-weighted automaton $\mathcal{Y}$ with a state space of size not greater than the one of $\mathcal{X}$, such that above inclusion holds. The state space of $\mathcal{Y}$ can be chosen as a basis for the underlying vector space of the epi-mono factorisation of the reachability map $\mathcal{X}(A^*) \rightarrow \overline{\textnormal{im}(\textnormal{obs}^{\dagger}_{\mathcal{X}})}^{\mathbb{Z}_2\textnormal{-Vect}}$.
	\end{proof}

\begin{oneshot}{corollary}{minimalityatomaton}
		The \'atomaton for $\mathcal{L}$ is size-minimal among non-deterministic automata $\mathcal{Y}$ accepting $\mathcal{L}$ with $\overline{\textnormal{im}(\obs^{\dag}_{\mathcal{Y}})}^{\textnormal{CSL}} = \overline{\textnormal{im}(\obs^{\dag}_{\mathcal{Y}})}^{\textnormal{CABA}}$.
\end{oneshot}
\begin{proof}
	Let $\alpha: \lambda^{\mathcal{H}} \rightarrow \lambda^{\mathcal{P}}$ be the distributive law homomorphism introduced in \Cref{alphapowersetneighbourhooddistrlaw}. As shown in \Cref{imgcabanfa}, the equality in above claim captures $\alpha$-closedness of $\mathcal{Y}$. By construction there exists a CSL-epimorphism $\textnormal{obs}_{\textnormal{exp}_{\mathcal{P}}(\mathcal{Y})}: \mathcal{P}Y \twoheadrightarrow \overline{\textnormal{im}(\obs^{\dag}_{\mathcal{Y}})}^{\textnormal{CSL}}$, which turns $Y$ into a $\mathcal{P}$-algebra generator for $B := \overline{\textnormal{im}(\obs^{\dag}_{\mathcal{Y}})}^{\textnormal{CSL}} = \overline{\textnormal{im}(\obs^{\dag}_{\mathcal{Y}})}^{\textnormal{CABA}}$. As for CABAs join-irreducibles and atoms coincide, the size-minimality of join-irreducibles in \Cref{joinirreducstateminimal} thus implies
	$
	\vert \textnormal{At}(B) \vert \leq \vert Y \vert
	$. Since $\overline{\textnormal{Der}(\mathcal{L})}^{\textnormal{CABA}}$ underlies the minimal $\alpha$-closed pointed $\lambda^{\mathcal{H}}$-bialgebra accepting $\mathcal{L}$, we have $ \overline{\textnormal{Der}(\mathcal{L})}^{\textnormal{CABA}} \subseteq B$ by \Cref{minimalitytheorem}, where we use \Cref{atomsstateminimalgenerator} and \Cref{obsgenerated}. By \Cref{atomscaba}, the former implies $\vert \textnormal{At}(\overline{\textnormal{Der}(\mathcal{L})}^{\textnormal{CABA}}) \vert \leq \vert \textnormal{At}(B) \vert$. Consequently we can deduce $\vert  \textnormal{At}(\overline{\textnormal{Der}(\mathcal{L})}^{\textnormal{CABA}}) \vert \leq \vert Y \vert$, which  shows the claim.  
\end{proof}

\begin{oneshot}{corollary}{distromatonminimal}
		The distromaton for $\mathcal{L}$ is size-minimal among non-deterministic automata $\mathcal{Y}$ accepting $\mathcal{L}$ with $ \overline{\textnormal{im}(\obs^{\dag}_{\mathcal{Y}})}^{\textnormal{CSL}} = \overline{\textnormal{im}(\obs^{\dag}_{\mathcal{Y}})}^{\textnormal{CDL}}$.
\end{oneshot}
\begin{proof}
	Analogous to the proof of \Cref{minimalityatomaton}. Specific to this case are \Cref{powersetalgebra}, \Cref{monotoneneighbourhoodoutputalgebra}, \Cref{neighbourhoodpowersetmorphism}, \Cref{joinirreducstateminimal}, \Cref{join-irred-dlattice}, and \Cref{imgdistromaton}.
\end{proof}

\begin{oneshot}{corollary}{corminimalxorcaba}
		The minimal xor-CABA automaton for $\mathcal{L}$ is size-minimal among $\mathbb{Z}_2$-weighted automata $\mathcal{Y}$ accepting $\mathcal{L}$ with $ \overline{\textnormal{im}(\obs^{\dag}_{\mathcal{Y}})}^{\mathbb{Z}_2\textnormal{-Vect}} = \overline{\textnormal{im}(\obs^{\dag}_{\mathcal{Y}})}^{\textnormal{CABA}}$.
\end{oneshot}
\begin{proof}
	Analogous to the proof of \Cref{minimalityatomaton}. Specific to this case are \Cref{neighbourhoodalgebra}, \Cref{xoroutputalgebra}, \Cref{alphaxorneighbourhooddistrlaw}, \Cref{xorbasisstateminimal}, \Cref{imgxorcaba} and the observation that if $A \subseteq B$ is a sub-vector space of a finite vector space $B$, then $\textnormal{dim}(A) \leq \textnormal{dim}(B)$.
\end{proof}

\begin{oneshot}{corollary}{sizecomparison}
	\begin{itemize}
		\item If $\overline{\textnormal{Der}(\mathcal{L})}^{\mathbb{Z}_2\textnormal{-Vect}} = \overline{\textnormal{Der}(\mathcal{L})}^{\textnormal{CABA}}$, then the minimal xor automaton and the minimal xor-CABA automaton for $\mathcal{L}$ are of the same size.
		\item If $\overline{\textnormal{Der}(\mathcal{L})}^{\textnormal{CSL}} = \overline{\textnormal{Der}(\mathcal{L})}^{\textnormal{CDL}}$, then the canonical RFSA and the distromaton for $\mathcal{L}$ are of the same size.
		\item If $\overline{\textnormal{Der}(\mathcal{L})}^{\textnormal{CSL}} = \overline{\textnormal{Der}(\mathcal{L})}^{\textnormal{CABA}}$, then the canonical RFSA and the \'atomaton for $\mathcal{L}$ are of the same size. 
	\end{itemize}
\end{oneshot}
\begin{proof}
	\begin{itemize}
		\item By \Cref{minimalxor} the minimal xor automaton $\mathcal{X}$ is of size not greater than the minimal xor-CABA automaton $\mathcal{Y}$. Conversely, we find
		\begin{align*}
			\overline{\textnormal{Der}(\mathcal{L})}^{\textnormal{CABA}} &= \overline{\textnormal{Der}(\mathcal{L})}^{\mathbb{Z}_2\textnormal{-Vect}} & \textnormal{(Assumption)} \\
			&= \textnormal{im}(\textnormal{obs}_{\textnormal{exp}_{\mathcal{R}}(\mathcal{X})}) & \textnormal{(\Cref{obsgenerated})}  \\
			&= \overline{\textnormal{im}(\textnormal{obs}^{\dagger}_{\mathcal{X}})}^{\mathbb{Z}_2\textnormal{-Vect}} & \textnormal{(\Cref{imgxorcaba})},
		\end{align*}
		which can be used to show $\overline{\textnormal{im}(\textnormal{obs}^{\dagger}_{\mathcal{X}})}^{\mathbb{Z}_2\textnormal{-Vect}} = \overline{\textnormal{im}(\textnormal{obs}^{\dagger}_{\mathcal{X}})}^{\textnormal{CABA}}$. By \Cref{corminimalxorcaba} the latter implies that $\mathcal{Y}$ is of size not greater than $\mathcal{X}$, which shows the claim.
	\item Let $\mathcal{X}$ denote the canonical RFSA and $\mathcal{Y}$ the distromaton. On the one hand we find
	\begin{align*}
			\overline{\textnormal{Der}(\mathcal{L})}^{\textnormal{CSL}} &= \overline{\textnormal{Der}(\mathcal{L})}^{\textnormal{CDL}} & \textnormal{(Assumption)} \\
			&= \textnormal{im}(\textnormal{obs}_{\textnormal{exp}_{\mathcal{P}}(\mathcal{Y})}) & \textnormal{(\Cref{obsgenerated})}  \\
			&= \overline{\textnormal{im}(\textnormal{obs}^{\dagger}_{\mathcal{Y}})}^{\textnormal{CSL}} & \textnormal{(\Cref{imgdistromaton})},
		\end{align*}	
		which by \Cref{canonicalrfsaminimal} implies that $\mathcal{X}$ is of size not greater than $\mathcal{Y}$. Conversely, we establish the equality
		\begin{align*}
			\overline{\textnormal{Der}(\mathcal{L})}^{\textnormal{CDL}} &= \overline{\textnormal{Der}(\mathcal{L})}^{\textnormal{CSL}} & \textnormal{(Assumption)} \\
			&= \textnormal{im}(\textnormal{obs}_{\textnormal{exp}_{\mathcal{P}}(\mathcal{X})}) & \textnormal{(\Cref{obsgenerated})}  \\
			&= \overline{\textnormal{im}(\textnormal{obs}^{\dagger}_{\mathcal{X}})}^{\textnormal{CSL}} & \textnormal{(\Cref{imgdistromaton})},
		\end{align*}	
		which can be used to show $\overline{\textnormal{im}(\textnormal{obs}^{\dagger}_{\mathcal{X}})}^{\textnormal{CSL}} = \overline{\textnormal{im}(\textnormal{obs}^{\dagger}_{\mathcal{X}})}^{\textnormal{CDL}}$. By \Cref{distromatonminimal} the latter implies that $\mathcal{Y}$ is of size not greater than $\mathcal{X}$, which shows the claim.
		\item The proof for the \'atomaton is analogous to the proof for the distromaton in the previous point.
	\end{itemize}
\end{proof}
\fi
\end{document}